\documentclass[aps,pra,reprint,amsmath,amssymb,longbibliography,superscriptaddress,noeprint]{revtex4-2}

\usepackage{bm}
\usepackage{amsmath,amssymb,amsfonts,amsthm,mathtools}
\usepackage{booktabs}
\usepackage{array}
\usepackage{xcolor}
\usepackage{microtype}
\usepackage{comment}
\usepackage{braket}
\usepackage[colorlinks=true,citecolor=blue!55!black,linkcolor=blue!55!black,urlcolor=blue!55!black]{hyperref}

\newtheorem{theorem}{Theorem}

\newtheorem{lemma}[theorem]{Lemma}
\newtheorem{corollary}[theorem]{Corollary}
\newtheorem{definition}[theorem]{Definition}

\newcommand{\cH}{\mathcal{H}}

\newcommand{\cD}{\mathcal{D}}
\newcommand{\cV}{\mathcal{V}}
\newcommand{\Tr}{\operatorname{Tr}}
\newcommand{\id}{\mathbb{I}}
\newcommand{\proj}[1]{\ket{#1} \bra{#1}}

\newcommand{\norm}[1]{\left\lVert #1\right\rVert}
\newcommand{\comm}[2]{\left[#1,#2\right]}

\newcommand{\poly}{\operatorname{poly}}

\begin{document}

\title{Quantum Complexity of Ancilla-Free Unitary Embeddings for Nonlinear Dynamics \\
via Generalized State-Dependent Double-Bracket Flows}

\author{Yuki Ito}
\email{yuki.itoh.osaka@gmail.com}
\affiliation{Graduate School of Engineering Science, The University of Osaka, 1-3 Machikaneyama, Toyonaka, Osaka 560-8531, Japan}

\author{Hideaki Hakoshima}
\affiliation{Center for Quantum Information and Quantum Biology, The University of Osaka, 1-2 Machikaneyama, Toyonaka, Osaka 560-0043, Japan}

\author{Keisuke Fujii}
\email{fujii.keisuke.es@osaka-u.ac.jp}
\affiliation{Graduate School of Engineering Science, The University of Osaka, 1-3 Machikaneyama, Toyonaka, Osaka 560-8531, Japan}
\affiliation{Center for Quantum Information and Quantum Biology, The University of Osaka, 1-2 Machikaneyama, Toyonaka, Osaka 560-0043, Japan}
\affiliation{RIKEN Center for Quantum Computing (RQC), Hirosawa 2-1, Wako, Saitama 351-0198, Japan}
\affiliation{Graduate School of Informatics, Kyoto University, Sakyo-ku, Kyoto, 606-8501, Japan}
\date{\today}

\begin{abstract}
Simulating nonlinear dynamics with quantum computers has gained increasing attention.
In general, such simulations require additional quantum resources because unitary quantum evolution is linear.
A fundamental question is how nonlinear dynamics can be embedded into fully coherent, ancilla-free unitary circuits and how the complexity of the dynamics governs the required quantum resources.
In this work, we generalize the ancilla-free double-bracket quantum algorithm for imaginary-time evolution by replacing its state-independent Hamiltonian with a state-dependent Hermitian operator.
Our framework recursively calls an initial state preparation oracle and its inverse, and prepares the target solution to any prescribed accuracy using a fully coherent, ancilla-free unitary embedding.
We relate the query cost to the complexity of the nonlinear dynamics, specifically their sensitivity to initial conditions.
We obtain query upper bounds of $\exp(O(T))$, $\exp(O(T^2))$, and $\exp(\exp(O(T)))$ when the distance between solutions contracts at least exponentially (contractive), does not increase (nonexpansive), or grows at most exponentially (expansive), respectively, where $T$ is the target evolution time.
For the discrete Gross--Pitaevskii equation, our ancilla-free double-bracket circuit achieves optimal worst-case query complexity $\Theta(e^{gT/2})$ over a specified family of single-qubit initial states, where $g>0$ is the nonlinearity strength.
These results connect the complexity of nonlinear dynamics to the query cost of coherent quantum simulation and provide a foundation for designing ancilla-free unitary embeddings with optimal query complexity.
\end{abstract}

\maketitle

\section{Introduction}
Differential equations are widely used to model various phenomena in
diverse fields, including
electromagnetics~\cite{Baumjohann1996SpacePlasma},
fluid dynamics~\cite{Burgers1948Turbulence},
ecology~\cite{Wangersky1978}, and
economics~\cite{black_pricing_1973, ankudinova_numerical_2008}.
Many differential equations of practical importance are nonlinear and
cannot be solved analytically, and numerical simulation is therefore
essential for investigating their behavior.
Simulating large-scale nonlinear systems at high accuracy can, however,
require considerable computational resources.

Quantum computers offer new computational capabilities, as exemplified by
quantum algorithms for prime factorization~\cite{shor_algorithms_1994},
matrix inversion~\cite{harrow_quantum_2009, childs_quantum_2017, costa_optimal_2022}, and
Hamiltonian simulation~\cite{feynman1982simulating,lloyd_universal_1996,LowChuang2017,LowChuang2019,Gilyen2019}.
These developments have motivated quantum approaches to differential
equations, including nonlinear ones~\cite{leyton_quantum_2008,
berry_high-order_2014,
montanaro_quantum_2016,
berry_quantum_2017,
childs_quantum_2020,
lloyd_quantum_2020,
fang_time-marching_2023,
an_quantum_2025,
berry_quantum_2024,
jin_quantum_2024,
an_linear_2023,
an_quantum_2026,
low_quantum_2026,
shang_design_2025,
jin_schrodingerization_2025,
low_optimal_2025,
gan2025provablyefficientquantumalgorithms,
Liu2021, Krovi2023}.
Because isolated quantum systems obey the linear Schr\"odinger equation,
however, nonlinear dynamics cannot in general be represented by
direct Hamiltonian evolution of the solution.
Instead, the nonlinearity must be embedded into a larger linear quantum
dynamics or reproduced through state-dependent coherent operations.
Such embeddings can incur additional costs in qubit count, oracle queries,
or circuit depth.
A representative approach is Carleman linearization, which embeds nonlinear
dynamics into an enlarged linear system by introducing higher-order tensor
powers of the solution~\cite{Liu2021,Krovi2023,LiuEtAl2023ReactionDiffusion,SuranaGnanasekaranSahai2023,BrustleWiebe2025,CostaEtAl2025Carleman,WuWangLi2025Carleman,Jennings2025Carleman,endo_divergence-free_2026,WangEtAl2026PivotCarleman}.
This raises the question of whether nonlinear dynamics can
instead be simulated in a fully coherent embedding without
such tensor powers or ancilla qubits.
A related question is what quantum resources such a simulation
would require.

There exists a fully coherent, ancilla-free method for simulating
certain nonlinear dynamics arising from the normalization of
linear evolution.
Specifically,
double-bracket quantum imaginary-time evolution (DB-QITE) 
coherently simulates the normalized imaginary-time evolution
$\ket{\psi(t)}=e^{-Ht}\ket{\psi(0)}/\norm{e^{-Ht}\ket{\psi(0)}}_2$
without ancilla qubits, where $H$ is a time- and state-independent Hamiltonian
and $\norm{\cdot}_2$ denotes the Euclidean norm~\cite{Gluza2026DBQITE}.
This evolution is
described by the well-studied double-bracket
flow~\cite{Bloch1990SteepestDescent,Moore1994NumericalGradient,Brockett1991DynamicalSystems,%
Smith1993GeometricOptimization,Helmke1994Optimization,Bloch1985Estimation,%
Bloch1992IntegrableGradientFlows,Brockett1989LeastSquares,Bloch1985LineFitting,%
Deift1983Eigenvalue,Chu1988Continuous,Wegner1994FlowEquations,%
Wegner2006Survey,Hastings2022DoubleBracket,Glazek1993Renormalization,%
Glazek1994Perturbative,Kehrein2006FlowEquation,Brockett1989SmoothSystems}
\begin{equation}
    \frac{d}{dt}\rho(t)
    =\comm{\comm{\rho(t)}{H}}{\rho(t)},
    \label{eq:fixed-db}
\end{equation}
where $\rho(t)=\proj{\psi(t)}$ is the pure-state projector.
The nonlinearity in this evolution arises solely from normalization.
A natural question is whether inherently nonlinear dynamics can be simulated using a fully coherent, ancilla-free approach analogous to DB-QITE.

In this work,
we develop a fully coherent, ancilla-free unitary embedding framework for nonlinear dynamics based on state-dependent double-bracket flows.
Specifically,
we generalize DB-QITE by replacing 
the state-independent Hamiltonian $H$ in
Eq.~\eqref{eq:fixed-db} with a state-dependent Hermitian operator $G(\rho)$
and consider
\begin{equation}
    \frac{d}{dt}\rho(t)
    =\comm{\comm{\rho(t)}{G(\rho(t))}}{\rho(t)}.
    \label{eq:state-dependent-db}
\end{equation}
Under suitable boundedness and Lipschitz continuity conditions on $G(\rho)$,
together with access to an ancilla-free unitary approximating
$e^{iG(\rho)\theta}$ for $\theta\in\mathbb{R}$,
we construct a recursive unitary implementation.
Starting from an initial state preparation oracle and its inverse, 
we approximate the evolution up to time $T$
by successively applying one-step updates of duration $\tau$.
At each step, 
we recursively construct the next state preparation circuit
by coherently calling the current state preparation circuit and its inverse as subroutines.
This framework can apply to a nonlinear Schr\"odinger equation (NLSE), 
including the discrete Gross--Pitaevskii (GP) equation.

The central quantity in our complexity analysis is a trace-norm Lipschitz
bound $\Lambda(t)$, which bounds the trace-norm distance between two solutions at
time $t$ relative to their initial distance.
It quantifies whether the distance between solutions contracts
at least exponentially $(\Lambda(t)=e^{-\lambda t}, \lambda > 0)$,
does not increase $(\Lambda(t)=1)$,
or can grow at most exponentially $(\Lambda(t)=e^{\lambda t}, \lambda>0)$.
With fixed model parameters and
a constant number of subroutine calls per update,
we obtain upper bounds of $\exp(O(T/\epsilon))$,
$\exp(O(T^2/\epsilon))$, and $\exp(e^{O(T)}/\epsilon)$
on the total number of calls to the initial state preparation
oracle and its inverse
when $\Lambda(t)=e^{-\lambda t}, 1$, and $e^{\lambda t}$,
respectively, for $T\ge 1$ and target accuracy $0<\epsilon\le 1$.
For the first two forms of $\Lambda(t)$, 
we obtain tighter query bounds at fixed accuracy than 
the general doubly exponential bounds in $T$ obtained 
by adapting earlier DB-QITE constructions and error-accumulation analyses~\cite{Gluza2026DBQITE,Wright_2026_DBTFD}
to the present state-dependent problem.
Thus, even within the same ancilla-free embedding framework, 
the query upper bounds can change qualitatively 
depending on 
the stability properties of the nonlinear dynamics, as characterized by the trace-norm Lipschitz bound $\Lambda(t)$.

To examine how these stability-dependent query upper bounds compare with the optimal worst-case query complexity, we consider a single-qubit instance of the discrete GP equation with a specified initial state family under suitable oracle-access assumptions.
In this instance, we first obtain a worst-case initial state preparation oracle
query lower bound of $\Omega(e^{gT/2})$ at fixed $0<\epsilon<1$, where $g>0$
is the nonlinearity strength.
For the query upper bound, 
the proposed framework yields 
a bound of $\exp(\exp(O(T)))$ based on the stability
properties of the dynamics, at fixed $g$ and $\epsilon$.
By incorporating additional information about the solution trajectories,
we improve this bound to $\exp(O(T^2))$,
uniformly over the specified initial state family.
However, 
this upper bound does not establish optimality. To achieve the optimal query complexity in a fully coherent, ancilla-free setting,
we further construct a
recursive circuit tailored to the exact solutions, 
which is related to double-bracket one-step
propagators.
Its worst-case query upper bound coincides with the lower bound up to a constant
factor, establishing the optimal worst-case query complexity
$\Theta(e^{gT/2})$ at fixed $0<\epsilon<1$ for sufficiently large $gT$.
While previous work~\cite{BrustleWiebe2025} does not
establish that its algorithm attains the lower bound
up to a constant factor,
our recursive construction achieves 
the optimal query complexity
for the specified initial state family, even under the ancilla-free constraint.

Taken together, our results suggest that trace-norm stability
can serve as a complexity parameter for coherent nonlinear simulation
beyond normalized imaginary-time evolution.
They also lay the groundwork for designing ancilla-free unitary embeddings that achieve optimal query complexity.

The remainder of this paper is organized as follows.
Section~\ref{sec:proposed-framework} introduces the state-dependent
double-bracket dynamics and their general recursive ancilla-free implementation.
Section~\ref{sec:one-step-error-bounds} first derives one-step
error bounds for the general implementation and then presents
the NLSE-specific implementation together with its one-step
error bounds.
Section~\ref{sec:global-analysis} uses a trace-norm Lipschitz bound to derive
global error bounds and initial state preparation oracle query
upper bounds.
Section~\ref{sec:application} applies the framework to the discrete
GP equation and analyzes the query complexity for the specified single-qubit initial state family.
Section~\ref{sec:conclusion} summarizes the results and discusses
future directions.

\section{Nonlinear Double-Bracket Dynamics and General Ancilla-Free Implementation}
\label{sec:proposed-framework}
In this section,
we introduce the class of nonlinear differential equations considered in this work
and present a general framework for their ancilla-free implementation.

\subsection{Target Nonlinear Dynamics}
We specify the target nonlinear dynamics. 
Let $\cH$ be an $N$-dimensional Hilbert space and
$\mathcal{S}_{\rm pure} \coloneqq
\left\{
    \rho=\proj{\psi}
    \mid
    \ket{\psi}\in\cH,
    \ \braket{\psi|\psi}=1
\right\}.
$
Motivated by the DB-QITE dynamics in Eq.~\eqref{eq:fixed-db},
we replace the state-independent Hamiltonian $H$ with a
state-dependent Hermitian operator $G(\rho)$.
For an initial state $\rho(0)\in\mathcal{S}_{\rm pure}$, 
we consider the nonlinear double-bracket dynamics given by Eq.~\eqref{eq:state-dependent-db}.
As shown in Appendix~\ref{appendix:purity-preservation},
Eq.~\eqref{eq:state-dependent-db} preserves all eigenvalues of $\rho(t)$.
Consequently, $\rho(t)$ remains in $\mathcal{S}_{\rm pure}$ throughout
the evolution whenever $\rho(0)\in\mathcal{S}_{\rm pure}$.
We assume that there exist finite constants
$M_G,L_G\ge0$ such that, for all
$\rho,\sigma\in\mathcal{S}_{\rm pure}$,
\begin{equation}
    \norm{G(\rho)}_{\rm op}
    \le
    M_G,
    \label{eq:G-sup-on-pure}
\end{equation}
and
\begin{equation}
    \norm{G(\rho)-G(\sigma)}_{\rm op}
    \le
    L_G\norm{\rho-\sigma}_1,
    \label{eq:G-Lipschitz-bound}
\end{equation}
where $\norm{\cdot}_{\rm op}$ and $\norm{\cdot}_1$
denote the operator norm induced by the Euclidean norm
and the trace norm, respectively.
For any linear operator $X$ on $\cH$, the trace norm
is defined by $\norm{X}_1 \coloneqq \Tr\sqrt{X^\dagger X}$.
Eqs.~\eqref{eq:G-sup-on-pure} and \eqref{eq:G-Lipschitz-bound} represent boundedness and Lipschitz continuity conditions on $G(\rho)$, respectively.
We write $\varphi_t(\rho)$ for the exact solution of Eq.~\eqref{eq:state-dependent-db} at time $t$ with initial state $\rho$.
Lemma~\ref{lem:generic-exponential-solution-bound} establishes an upper bound
that allows at most exponential expansion of the trace-norm distance
between two solutions.

\begin{lemma}
\label{lem:generic-exponential-solution-bound}
Under Eqs.~\eqref{eq:G-sup-on-pure} and
\eqref{eq:G-Lipschitz-bound}, define
\begin{equation}
    \lambda_{\rm exp}
    \coloneqq
    2(M_G+L_G).
\end{equation}
Then, for all $\rho,\sigma\in\mathcal{S}_{\rm pure}$ and $t\ge0$,
\begin{equation}
    \norm{
        \varphi_t(\rho)-\varphi_t(\sigma)
    }_1
    \le
    e^{\lambda_{\rm exp}t}
    \norm{\rho-\sigma}_1.
    \label{eq:generic-exponential-solution-bound}
\end{equation}
\end{lemma}

\begin{proof}
See Appendix~\ref{appendix:Lem-generic-exponential-solution-bound}.
\end{proof}

\noindent
The parameter $\lambda_{\rm exp}$ provides a general upper
bound on the exponential growth rate of the trace-norm
distance between solutions.
In the global error and query-complexity analysis of Sec.~\ref{sec:global-analysis}, 
this bound controls the amplification of perturbations in the initial state.

As a concrete example, 
we consider an NLSE of the form
\begin{equation}
    i\frac{d}{dt}\ket{\psi(t)}
    =
    H(\rho(t))\ket{\psi(t)},
    \label{eq:NLSE}
\end{equation}
where $H(\rho)$ is a Hermitian operator.
We assume that there exist finite constants $M_H,L_H\ge0$ such that,
for all $\rho,\sigma\in\mathcal{S}_{\rm pure}$,
\begin{equation}
    \norm{H(\rho)}_{\rm op}
    \le
    M_H,
    \label{eq:H-bound-on-pure}
\end{equation}
and
\begin{equation}
    \norm{H(\rho)-H(\sigma)}_{\rm op}
    \le
    L_H\norm{\rho-\sigma}_1.
    \label{eq:H-Lipschitz-bound}
\end{equation}
If
\begin{equation}
    G(\rho)
    =
    -i\comm{\rho}{H(\rho)},
    \label{eq:generator-NLSE}
\end{equation}
the density operator $\rho(t)=\proj{\psi(t)}$ associated with
Eq.~\eqref{eq:NLSE} satisfies Eq.~\eqref{eq:state-dependent-db}.
The commutator bounds for rank-one projectors show that
Eqs.~\eqref{eq:G-sup-on-pure} and
\eqref{eq:G-Lipschitz-bound} hold, for example, with
\begin{equation}
    M_G = M_H,
\end{equation}
and
\begin{equation}
    L_G = M_H+L_H.
\end{equation}

Another example is the nonlinear imaginary-time
evolution (NITE) introduced in Ref.~\cite{shi2026generalized}.
For a differentiable real-valued cost function $\mathcal{E}(\rho)$,
consider the choice
$G(\rho)=\nabla_{\rho}^{(\mathrm{HS})}\mathcal{E}(\rho)$,
where $\nabla_{\rho}^{(\mathrm{HS})}$ denotes the
Hilbert--Schmidt gradient.
With this choice, Eq.~\eqref{eq:state-dependent-db} reproduces the corresponding
NITE dynamics on pure states.
We now return to state-dependent Hermitian operators $G(\rho)$ satisfying
Eqs.~\eqref{eq:G-sup-on-pure} and~\eqref{eq:G-Lipschitz-bound}.

\subsection{General Recursive Ancilla-Free Implementation} 
\label{subsec:general-framework}
In this subsection, we present the proposed recursive ancilla-free
simulation framework.
With $\rho(t)=\proj{\psi(t)}$ and a suitable choice of global phase, 
Eq.~\eqref{eq:state-dependent-db} admits the state-vector representation
\begin{equation}
    \frac{d}{dt}\ket{\psi(t)}
    =
    \comm{\rho(t)}{G(\rho(t))}\ket{\psi(t)}.
\end{equation}
We assume access to an initial state preparation oracle $U_0$
satisfying $U_0\ket{0}=\ket{\psi(0)}$ and its inverse $U_0^\dagger$.
Here, $\ket{0}$ denotes the tensor-product zero state.
Our goal is to use $U_0$ and $U_0^\dagger$ to recursively construct
an ancilla-free circuit that prepares an approximation to
$\ket{\psi(T)}$ up to a global phase, where $T$ is the target
evolution time.
We measure the approximation error by the trace-norm distance.
To this end, we choose a step size $\tau>0$ such that
$M\coloneqq T/\tau$ is a positive integer.
As shown in Fig.~\ref{fig:recursive-implementation}, we simulate the dynamics up to the target evolution time $T$ by recursively applying one-step updates of duration $\tau$.
We first describe the one-step updates.

\begin{figure*}[t]
  \begin{center}
    \includegraphics[width=\linewidth]{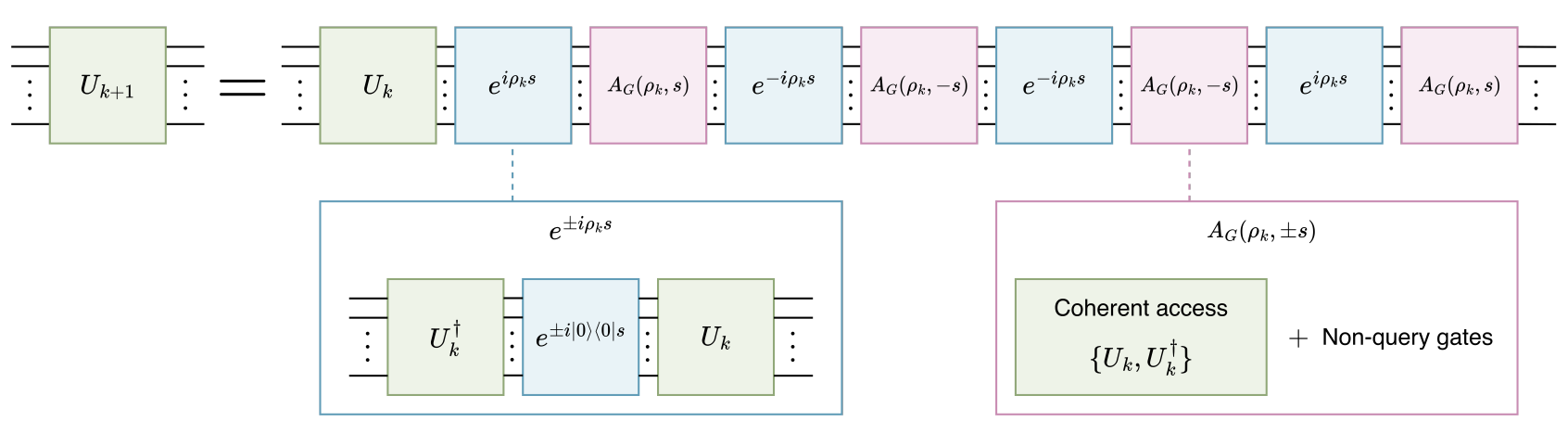}
    \caption{An overview of the recursive ancilla-free implementation of nonlinear
    state-dependent double-bracket dynamics in Eq.~\eqref{eq:state-dependent-db}.
    This circuit implements an approximate one-step update of duration
    $\tau$, recursively constructing $U_{k+1}$ from $U_k$ and $U_k^\dagger$.
    The state preparation circuit $U_k$ satisfies
    $U_k\ket{0}=\ket{\psi_k}$, where $\ket{\psi_k}$ approximates
    the solution at $t_k=k\tau$ up to a global phase.
    This implementation is based on the
    symmetric group-commutator formula, with $s=\sqrt{\tau/2}$.
    The projector exponentials are implemented as
    $e^{\pm i\rho_k s}=U_k e^{\pm i\proj{0}s}U_k^\dagger$,
    where $\rho_k=\proj{\psi_k}$.
    The ancilla-free unitaries $A_G(\rho_k,\pm s)$ approximate
    $e^{\pm iG(\rho_k)s}$
    and are implemented using coherent access to $U_k$ and $U_k^\dagger$,
    together with non-query gates.}
    \label{fig:recursive-implementation}
  \end{center}
\end{figure*}

For an input state $\rho$, 
we hold $\comm{\rho}{G(\rho)}$ fixed over one step and define the discretized one-step propagator by
\begin{equation}
    W(\rho,\tau)
    \coloneqq
    e^{\tau\comm{\rho}{G(\rho)}}.
    \label{eq:discretized-one-step-propagator}
\end{equation}
This gives the one-step approximation
\begin{equation} \label{eq:first-order-approximation}
    \ket{\psi(t+\tau)}
    = W(\rho(t),\tau)\ket{\psi(t)}+O(\tau^2).
\end{equation}
The corresponding discretized one-step map on density operators is
\begin{equation}
    \Phi_\tau(\rho)
    =
    W(\rho,\tau)\rho W(\rho,\tau)^\dagger.
\end{equation}
To approximate $W(\rho,\tau)$ without ancilla qubits, we use the symmetric
group-commutator formula analyzed in Sec.~\ref{subsec:one-step-error-general-implementation}:
\begin{align}
    W(\rho,\tau)
    ={}&e^{\tau\comm{\rho}{G(\rho)}} \\
    ={}&e^{iG(\rho)s}e^{i\rho s}e^{-iG(\rho)s}e^{-i\rho s} \\
    &\times e^{-iG(\rho)s}e^{-i\rho s}e^{iG(\rho)s}e^{i\rho s}
    +O(\tau^2), \notag
\end{align}
where $s=\sqrt{\tau/2}$.
We approximate $W(\rho,\tau)$ using state preparation circuits for the projector exponentials 
and ancilla-free approximations to the exponentials of $G(\rho)$.
We assume that, given coherent access to a state preparation circuit
for $\rho$ and its inverse, we can implement an ancilla-free unitary
$A_G(\rho,\theta)$ satisfying
\begin{equation}
    \norm{
        A_G(\rho,\theta)-e^{iG(\rho)\theta}
    }_{\rm op}
    \le
    \epsilon_G(\theta).
\end{equation}
The ancilla-free circuit implementing $A_G(\rho,\theta)$
and its approximation parameters may be chosen separately
for each $\theta$.
We require $\epsilon_G(\pm s)=O(\tau^2)$, as specified in
Sec.~\ref{subsec:one-step-error-general-implementation}.
Replacing the exponentials of $G(\rho)$ by these unitaries gives
\begin{equation}
    \label{eq:symmetric-gc}
    \begin{split}
        V_{\rm Gen}(\rho,\tau)
        \coloneqq{}&
        A_G(\rho,s)e^{i\rho s}A_G(\rho,-s)e^{-i\rho s} \\
        &\times
        A_G(\rho,-s)e^{-i\rho s}A_G(\rho,s)e^{i\rho s}.
    \end{split}
\end{equation}
The implemented one-step map is
\begin{equation}
    \widetilde{\Phi}^{({\rm Gen})}_\tau(\rho)
    =
    V_{\rm Gen}(\rho,\tau)\rho V_{\rm Gen}(\rho,\tau)^\dagger.
\end{equation}

Using $V_{\rm Gen}(\rho,\tau)$, we define states $\ket{\psi_k}$
that approximate $\ket{\psi(t_k)}$ with $t_k\coloneqq k\tau$ and
$k\in\{0,1,\ldots,M\}$.
Starting from $\ket{\psi_0}=\ket{\psi(0)}$, we define
\begin{equation}
    \ket{\psi_{k+1}}
    \coloneqq
    V_{\rm Gen}(\rho_k,\tau)\ket{\psi_k},
\end{equation}
where $\rho_k\coloneqq\proj{\psi_k}$ and $k\in\{0,1,\ldots,M-1\}$.
We recursively construct unitaries $U_k$ satisfying
$U_k\ket{0}=\ket{\psi_k}$ via
\begin{equation}
    \label{eq:recursive-relation}
    U_{k+1}
    \coloneqq
    V_{\rm Gen}(\rho_k,\tau)U_k.
\end{equation}
At step $k$, $U_k$ and $U_k^\dagger$ provide the state preparation
access required to implement $A_G(\rho_k,\pm s)$.
The projector exponentials are implemented using the identity
$e^{i\rho_k\theta}=U_ke^{i\proj{0}\theta}U_k^\dagger$
for any $\theta\in\mathbb{R}$,
where the unitary $e^{i\proj{0}\theta}$ can be realized exactly and efficiently without ancilla qubits~\cite{barenco_elementary_1995}.
Thus, $U_{k+1}$ can be constructed from $U_k$ and $U_k^\dagger$,
and ultimately from $U_0$ and $U_0^\dagger$.
The resulting circuit $U_M$ prepares $\ket{\psi_M}$ as an
approximation to the target state $\ket{\psi(T)}$.

We finally determine the query complexity with respect to the
initial state preparation oracle $U_0$ and its inverse $U_0^\dagger$.
For each $k\in\{0,1,\ldots,M-1\}$, let $a_{G,k+1}$ be a common upper
bound on the total number of calls to $U_k$ and $U_k^\dagger$ needed
to implement either $A_G(\rho_k,s)$ or $A_G(\rho_k,-s)$ at the
required accuracy $\epsilon_G(\pm s)=O(\tau^2)$.
We treat the implementation as a black box whose query cost may
depend on $\tau$.
The four $A_G(\rho_k,\pm s)$ factors in $V_{\rm Gen}(\rho_k,\tau)$
require at most $4a_{G,k+1}$ calls to $U_k$ and $U_k^\dagger$,
while the four projector exponentials require eight.
Thus, $V_{\rm Gen}(\rho_k,\tau)$ requires at most $8+4a_{G,k+1}$
calls, and including the rightmost $U_k$ in
Eq.~\eqref{eq:recursive-relation} gives at most $9+4a_{G,k+1}$
calls for $U_{k+1}$.
Let $Q_k$ denote the total number of calls to $U_0$ and $U_0^\dagger$
in the recursively expanded implementation of $U_k$, with $Q_0=1$.
Recursive expansion gives
\begin{equation}
    Q_k
    \le
    \prod_{j=1}^k(9+4a_{G,j}),
    \qquad
    k\in\{1,2,\ldots,M\}.
\end{equation}
If there exists a uniform upper bound $a_G$ on the total number
of calls to $U_{j-1}$ and $U_{j-1}^\dagger$ required to implement
$A_G(\rho_{j-1},\pm s)$ for either sign at the required accuracy,
such that $a_{G,j}\le a_G$ for all
$j\in\{1,2,\ldots,M\}$, then
\begin{equation}
    Q_k
    \le
    (9+4a_G)^k
    =
    \exp\left(O\left(k\log(9+4a_G)\right)\right).
\end{equation}

In the next section, 
we bound the one-step discretization error between the exact solution $\varphi_\tau(\rho)$ after time $\tau$
and the corresponding state $\Phi_\tau(\rho)$ obtained by a one-step discretization,
as well as the one-step implementation error associated with $V_{\rm Gen}(\rho,\tau)$. 
We also present an NLSE-specific circuit with two reflections about the current state in place of four projector exponentials.

\section{One-Step Error Bounds}
\label{sec:one-step-error-bounds}
To bound the global error, measured by the trace-norm
distance between the exact solution and the state produced by the proposed framework 
at the target evolution time $T$,
we first derive one-step error bounds for the general
implementation.
We then use the NLSE-specific form of the state-dependent
Hermitian operator $G(\rho)$ 
to construct a simpler ancilla-free implementation
and bound its one-step error.

\subsection{A One-Step Error Bound for the General Recursive Ancilla-Free Implementation}
\label{subsec:one-step-error-general-implementation}
In this subsection, we bound the one-step trace-norm error of the
general recursive ancilla-free implementation introduced in Sec.~\ref{subsec:general-framework}, 
which includes the NLSE in Eq.~\eqref{eq:NLSE}.
We bound the discretization and implementation errors separately
and combine these bounds to obtain an $O(\tau^2)$ bound,
consistent with the first-order approximation in Eq.~\eqref{eq:first-order-approximation}.

For $\rho\in\mathcal{S}_{\rm pure}$, the one-step error of the
general implementation satisfies
\begin{equation}
    \begin{split}
        &\norm{
            \widetilde{\Phi}^{({\rm Gen})}_\tau(\rho)
            -\varphi_\tau(\rho)
        }_1
        \\
        \le{}&
        \norm{
            \widetilde{\Phi}^{({\rm Gen})}_\tau(\rho)
            -\Phi_\tau(\rho)
        }_1
        +
        \norm{
            \Phi_\tau(\rho)-\varphi_\tau(\rho)
        }_1.
    \end{split}
\end{equation}
We first bound the one-step discretization error in the second term 
and then the one-step implementation error in the first term.

\begin{lemma}[One-step discretization error bound]
\label{lem:discretization-error}
For every $\rho\in\mathcal{S}_{\rm pure}$ and $0<\tau\le1$,
\begin{equation}
    \norm{
        \Phi_\tau(\rho)-\varphi_\tau(\rho)
    }_1
    \le
    C_{\rm disc}\tau^2,
    \label{eq:one-step-discretization-error}
\end{equation}
where
\begin{equation}
    C_{\rm disc}
    \coloneqq
    M_G\lambda_{\rm exp}+2M_G^2.
\end{equation}
\end{lemma}

\begin{proof}
See Appendix~\ref{appendix:discretization-error}.
\end{proof}

We next derive an $O(\tau^2)$ bound on the one-step implementation error of $V_{\rm Gen}(\rho,\tau)$
in Eq.~\eqref{eq:symmetric-gc} relative to $W(\rho,\tau)$,
so that we can obtain an $O(\tau^2)$ bound on the total one-step error.
To this end, with $s=\sqrt{\tau/2}$, we assume that
\begin{equation}
    \epsilon_G(\pm s)\le c_G\tau^2,
\end{equation}
for all $0<\tau\le1$, where $c_G>0$ is independent of $\rho$ and $\tau$.

\begin{lemma}[One-step implementation error bound]
\label{lem:one-step-implementation-error}
For every $\rho\in\mathcal{S}_{\rm pure}$ and $0<\tau\le1$,
\begin{equation}
    \norm{
        V_{\rm Gen}(\rho,\tau)-W(\rho,\tau)
    }_{\rm op}
    \le
    C_{\rm GC}\tau^2,
    \label{eq:symmetric-gc-error-bound}
\end{equation}
where
\begin{equation}
    C_{\rm GC}
    \coloneqq
    \frac{8}{3}(1+M_G)^4
    +\frac{1}{2}M_G^2
    +4c_G.
\end{equation}
\end{lemma}

\begin{proof}
See Appendix~\ref{appendix:proof-one-step-implementation-error}.
\end{proof}

\noindent
Lemma~\ref{lem:one-step-implementation-error} gives the implementation error bound
\begin{equation}
    \norm{
        \widetilde{\Phi}^{({\rm Gen})}_\tau(\rho)
        -\Phi_\tau(\rho)
    }_1
    \le
    C_{\rm Gen}\tau^2,
\end{equation}
where $C_{\rm Gen} \coloneqq 2C_{\rm GC}$.
Combining this bound with Eq.~\eqref{eq:one-step-discretization-error}
yields
\begin{equation}
    \norm{
        \widetilde{\Phi}^{({\rm Gen})}_\tau(\rho)
        -\varphi_\tau(\rho)
    }_1
    \le
    (C_{\rm Gen}+C_{\rm disc})\tau^2.
\end{equation}

\subsection{An NLSE-Specific One-Step Error Bound}
\label{subsec:one-step-error-NLSE}
We now focus on the NLSE in Eq.~\eqref{eq:NLSE},
for which the state-dependent Hermitian operator takes
the commutator form $G(\rho)=-i\comm{\rho}{H(\rho)}$,
as given in Eq.~\eqref{eq:generator-NLSE}.
This structure allows us to construct a simpler ancilla-free
circuit that approximates the same one-step propagator
$W(\rho,\tau)$.
The circuit uses exponentials of $H(\rho)$ in place of
those of $G(\rho)$, together with reflections about
the current state.
We then bound the one-step implementation error in trace norm. 
Combining this estimate with the discretization error bound yields a bound on the total one-step error.

Under Eqs.~\eqref{eq:H-bound-on-pure} and~\eqref{eq:H-Lipschitz-bound},
we assume that, given coherent access to a state preparation
circuit for $\rho$ and its inverse, we can implement an
ancilla-free unitary $A_H(\rho,\theta)$ satisfying
\begin{equation}
    \norm{
        A_H(\rho,\theta)-e^{iH(\rho)\theta}
    }_{\rm op}
    \le
    \epsilon_H(\theta),
\end{equation}
where
\begin{equation}
    \epsilon_H(\theta)
    \le
    c_H\theta^2
\end{equation}
for all $0<\lvert\theta\rvert\le1$, 
with a constant $c_H>0$
independent of $\rho$ and $\theta$.
The ancilla-free circuit implementing $A_H(\rho,\theta)$
and its approximation parameters may be chosen separately
for each $\theta$.

Define the reflection associated with $\rho$ by $R_\rho = \id-2\rho$.
Using Eq.~\eqref{eq:generator-NLSE} and $\rho^2=\rho$, we obtain
\begin{equation}
    i\comm{\rho}{G(\rho)}
    =
    \frac{1}{2}
    \left(
        H(\rho)-R_\rho H(\rho)R_\rho
    \right).
\end{equation}
Hence
\begin{equation}
    W(\rho,\tau)
    =
    \exp\left[
        -\frac{i\tau}{2}
        \left(
            H(\rho)-R_\rho H(\rho)R_\rho
        \right)
    \right].
\end{equation}
Applying the first-order Lie--Trotter formula and replacing the
exponentials of $H(\rho)$ by $A_H(\rho,\pm\tau/2)$, we approximate
$W(\rho,\tau)$ by
\begin{equation}
    V_{\rm NLSE}(\rho,\tau)
    \coloneqq
    A_H(\rho,-\tau/2)R_\rho
    A_H(\rho,\tau/2)R_\rho.
    \label{eq:NLSE-one-step}
\end{equation}
For $\rho_k=U_k\proj{0}U_k^\dagger$, the reflection is implemented as
\begin{equation}
    R_{\rho_k}
    =
    U_k(\id-2\proj{0})U_k^\dagger.
\end{equation}
Thus, $V_{\rm NLSE}$ can be implemented without ancilla qubits, using
two reflections about the current state and two $A_H$ factors in place of the
four projector exponentials and four $A_G$ factors in $V_{\rm Gen}$.
Let $a_{H,k}$
be an upper bound on the total number of calls to $U_{k-1}$ and
$U_{k-1}^\dagger$ required to implement either
$A_H(\rho_{k-1},-\tau/2)$ or $A_H(\rho_{k-1},\tau/2)$.
Replacing $V_{\rm Gen}$ by $V_{\rm NLSE}$ in
Eq.~\eqref{eq:recursive-relation} gives the recursion factor
$5+2a_{H,k}$, compared with $9+4a_{G,k}$ for the general circuit.
The relative query costs therefore depend on the costs of implementing
$A_H$ and $A_G$ at their respective required accuracies.

For every $\rho\in\mathcal{S}_{\rm pure}$ and $0<\tau\le1$,
the first-order Lie--Trotter estimate and the bound on
$\epsilon_H(\theta)$ give
\begin{equation}
    \norm{
        V_{\rm NLSE}(\rho,\tau)-W(\rho,\tau)
    }_{\rm op}
    \le
    \left(
        \frac{c_H}{2}
        +\frac{M_H^2}{4}
    \right)\tau^2.
\end{equation}
Consequently, for
\begin{equation}
    \widetilde{\Phi}^{({\rm NLSE})}_\tau(\rho)
    \coloneqq
    V_{\rm NLSE}(\rho,\tau)\rho
    V_{\rm NLSE}(\rho,\tau)^\dagger,
\end{equation}
we have the implementation error bound
\begin{equation}
    \norm{
        \widetilde{\Phi}^{({\rm NLSE})}_\tau(\rho)
        -\Phi_\tau(\rho)
    }_1
    \le
    C_{\rm NLSE}\tau^2,
\end{equation}
where $C_{\rm NLSE} \coloneqq c_H+ M_H^2/2$.
Combining this implementation error bound with
Eq.~\eqref{eq:one-step-discretization-error} yields
\begin{equation}
    \norm{
        \widetilde{\Phi}^{({\rm NLSE})}_\tau(\rho)
        -\varphi_\tau(\rho)
    }_1
    \le
    (C_{\rm NLSE}+C_{\rm disc})\tau^2.
\end{equation}
In Sec.~\ref{sec:global-analysis}, we develop a unified analysis
of the propagation of one-step errors for both the general
and NLSE-specific implementations, yielding global error
bounds and initial state preparation oracle query upper bounds.

\section{Global Error Bounds and Query Complexity of the Initial State Preparation Oracle}
\label{sec:global-analysis}
In the preceding section,
we derived one-step error
bounds for both the general and NLSE-specific implementations.
In this section, we develop a unified analysis of global error and
initial state preparation oracle query complexity that applies to both implementations.
The key parameter is a trace-norm Lipschitz bound $\Lambda(t)$,
which bounds the sensitivity of the nonlinear evolution
to perturbations in the initial state.
We use this bound to control the propagation of one-step errors
and obtain a global error bound at the target evolution time $T$.
We then determine a sufficient number of time steps $M$ and the
resulting query complexity with respect to the initial state preparation
oracle $U_0$ and its inverse $U_0^\dagger$.

\subsection{Global Error}
\label{subsec:global-error}
We first bound the accumulation of one-step errors over the simulation.
The Lipschitz bound $\Lambda(t)$ defined below controls 
the amplification of perturbations in the initial state, measured in the trace norm,
and is used to determine a sufficient number of
recursive steps for achieving the target accuracy.
We allow the Lipschitz bound to hold on a subset
$\Omega\subset\mathcal{S}_{\rm pure}$ containing the exact
trajectory and the states generated by the algorithm.

\begin{definition}[Trace-norm Lipschitz bound]
\label{def:Lipschitz-bound}
Let $\Omega\subset\mathcal{S}_{\rm pure}$.
$\varphi_t(\sigma)$ denotes the exact solution of
Eq.~\eqref{eq:state-dependent-db} at time $t$
with initial state $\sigma \in \mathcal{S}_{\rm pure}$.
A function $\Lambda:[0,\infty)\to[0,\infty)$ is a trace-norm
Lipschitz bound for $\varphi_t$ on $\Omega$ if
\begin{equation}
    \norm{\varphi_t(\sigma_1)-\varphi_t(\sigma_2)}_1
    \le
    \Lambda(t)\norm{\sigma_1-\sigma_2}_1
    \label{eq:Lipschitz-bound}
\end{equation}
for all $\sigma_1,\sigma_2\in\Omega$ and $t\ge0$.
\end{definition}

The following global error estimate applies to a general
implemented one-step map $\widetilde\Phi_\tau$, including both
implementations in Sec.~\ref{sec:one-step-error-bounds}.

\begin{theorem}[Global error bound]
\label{thm:global-error}
Fix $T>0$ and $\rho_0\in\mathcal{S}_{\rm pure}$, and let
$\rho(t)=\varphi_t(\rho_0)$ be the corresponding solution of
Eq.~\eqref{eq:state-dependent-db}.
Choose a positive integer $M$ such that $\tau=T/M\le1$.
For the implemented one-step map $\widetilde\Phi_\tau$, define
\begin{equation}
    \rho_{k+1}=\widetilde\Phi_\tau(\rho_k)
    \qquad (k \in \{0,1,\dots,M-1\}).
\end{equation}
Suppose that $\Omega\subset\mathcal{S}_{\rm pure}$ satisfies
\begin{equation}
    \left\{\rho(t) \mid t\in[0,T]\right\}
    \cup
    \left\{\rho_k \mid k \in \{0,1,\dots,M\}\right\}
    \subset\Omega,
    \label{eq:Omega-containing-trajectories}
\end{equation}
and that $\varphi_t$ admits a Lipschitz bound $\Lambda(t)$ on
$\Omega$ as in Definition~\ref{def:Lipschitz-bound}.
Assume also that, for every $\sigma\in\Omega$ and $0<\tau\le1$,
\begin{equation}
    \norm{\widetilde\Phi_\tau(\sigma)-\varphi_\tau(\sigma)}_1
    \le C_{\rm step}\tau^2,
    \label{eq:exact-solution-difference}
\end{equation}
where $C_{\rm step}>0$ is independent of $\tau$ and $\sigma$.
Then
\begin{align}
    \norm{\rho_M-\rho(T)}_1
    &\le C_{\rm step}\tau^2
    \sum_{j=0}^{M-1}\Lambda(\tau)^j.
    \label{eq:global-error-exact} 
\end{align}
\end{theorem}

\begin{proof}
Let $e_k=\norm{\rho_k-\rho(t_k)}_1$, where $t_k=k\tau$.
Since $\rho(t_{k+1})=\varphi_\tau(\rho(t_k))$, the triangle
inequality and Eqs.~\eqref{eq:Lipschitz-bound}
and~\eqref{eq:exact-solution-difference} give
\begin{equation}
    \begin{split}
        e_{k+1}
        &\le
        \norm{\widetilde\Phi_\tau(\rho_k)
            -\varphi_\tau(\rho_k)}_1\\
        &\quad+
        \norm{\varphi_\tau(\rho_k)
            -\varphi_\tau(\rho(t_k))}_1\\
        &\le C_{\rm step}\tau^2+\Lambda(\tau)e_k.
    \end{split}
\end{equation}
Iterating this inequality with $e_0=0$ proves
Eq.~\eqref{eq:global-error-exact}.
\end{proof}

\noindent
The bounds in Sec.~\ref{sec:one-step-error-bounds} allow us to take
$\widetilde\Phi_\tau=\widetilde\Phi^{({\rm Gen})}_\tau$ and
$C_{\rm step}=C_{\rm Gen}+C_{\rm disc}$ for the general implementation, or
$\widetilde\Phi_\tau=\widetilde\Phi^{({\rm NLSE})}_\tau$ and
$C_{\rm step}=C_{\rm NLSE}+C_{\rm disc}$ for the NLSE-specific implementation.

\subsection{Query Complexity of the Initial State Preparation Oracle}
\label{subsec:initial-state-query-complexity}
We now combine the global error bound with the recursive circuit
construction to bound the query complexity with respect to the
initial state preparation oracle $U_0$ and its inverse $U_0^\dagger$.
We consider the following three cases:
the exponential contraction bound $\Lambda(t)=e^{-\lambda t}$,
the nonexpansive bound $\Lambda(t)=1$,
and the exponential expansion bound $\Lambda(t)=e^{\lambda t}$,
with $\lambda>0$ in the exponential cases.
Theorem~\ref{thm:query-complexity-Lipschitz-bounds},
one of the main results of this work,
provides upper bounds on the query complexity
in each of these three cases.

\begin{theorem}[Query complexity under different Lipschitz bounds]
\label{thm:query-complexity-Lipschitz-bounds}
Fix a target evolution time $T>0$ and an error tolerance $\epsilon>0$.
Assume that the hypotheses of Theorem~\ref{thm:global-error}
hold for the choices of $M$ specified below, and adopt its notation
with $\tau=T/M$.
Let the implemented one-step map be
\begin{equation}
    \widetilde\Phi_\tau(\rho)
    =V(\rho,\tau)\rho V(\rho,\tau)^\dagger,
\end{equation}
where $V(\rho,\tau)$ is unitary.
Let $U_0$ be an initial state preparation oracle satisfying
$U_0\ket{0}=\ket{\psi(0)}$, with $\rho_0=\proj{\psi(0)}$.
Define $U_{k+1}=V(\rho_k,\tau)U_k$
for $k\in\{0,1,\ldots,M-1\}$.
Suppose that each $V(\rho_k,\tau)$ admits an ancilla-free
implementation using at most $b$ calls in total to
$U_k$ and $U_k^\dagger$, uniformly in $k$.
Then $U_M$ prepares $\ket{\psi_M}=U_M\ket{0}$ with
$\rho_M=\proj{\psi_M}$ satisfying
$\norm{\rho_M-\rho(T)}_1\le\epsilon$ when
\begin{equation}
    M=
    \begin{dcases}
        \left\lceil\max\left\{
            1,T,\lambda T,
            \frac{2C_{\rm step}T}{\lambda\epsilon}
        \right\}\right\rceil
        & (\Lambda(t)=e^{-\lambda t}),\\
        \left\lceil\max\left\{
            1,T,\frac{C_{\rm step}T^2}{\epsilon}
        \right\}\right\rceil
        & (\Lambda(t)=1),\\
        \left\lceil\max\left\{
            1,T,\lambda T,
            \frac{C_{\rm step}Te^{\lambda T}}
                 {\lambda\epsilon}
        \right\}\right\rceil
        & (\Lambda(t)=e^{\lambda t}),
    \end{dcases}
    \label{eq:step-complexity-stability}
\end{equation}
where $\lambda>0$ in the exponential cases.
The total number $Q_M$ of queries to $U_0$ and $U_0^\dagger$
used to implement $U_M$ satisfies
\begin{equation}
    Q_M
    \le
    (1+b)^M
    =
    \exp\left(M\log(1+b)\right),
    \label{eq:query-complexity-stability}
\end{equation}
where $b$ may depend on the chosen step size and accuracy.
In particular, for the two implementations, we may take
\begin{equation}
    b=
    \begin{cases}
        8+4a_G, & \text{general implementation},\\
        4+2a_H, & \text{NLSE-specific implementation},
    \end{cases}
\end{equation}
where $a_G$ and $a_H$ uniformly bound the total numbers of calls
to $U_k$ and $U_k^\dagger$ required by
$A_G(\rho_k,\pm\sqrt{\tau/2})$ and
$A_H(\rho_k,\pm\tau/2)$, respectively, for either sign at the
accuracies specified in Sec.~\ref{sec:one-step-error-bounds}.
\end{theorem}

\begin{proof}
The definitions give $\rho_0=\rho(0)$ and
$\rho_{k+1}=\widetilde\Phi_\tau(\rho_k)$.
For $0<\lambda\tau\le1$, we have
$1-e^{-\lambda\tau}\ge\lambda\tau/2$, while
$e^{\lambda\tau}-1\ge\lambda\tau$ holds for all $\tau>0$.
Theorem~\ref{thm:global-error} therefore gives
\begin{equation}
    \norm{\rho_M-\rho(T)}_1
    \le
    \begin{dcases}
        \dfrac{2C_{\rm step}}{\lambda}\tau
        & (\Lambda(t)=e^{-\lambda t}),\\
        C_{\rm step}T\tau
        & (\Lambda(t)=1),\\
        \dfrac{C_{\rm step}e^{\lambda T}}{\lambda}\tau
        & (\Lambda(t)=e^{\lambda t}).
    \end{dcases}
\end{equation}
In each maximum in Eq.~\eqref{eq:step-complexity-stability},
the entries $1$ and $T$ ensure $M\ge1$ and $\tau\le1$,
respectively.
The entry $\lambda T$ ensures $\lambda\tau\le1$ in both
exponential cases.
The final entry makes the corresponding error bound at most
$\epsilon$.
The recursion $U_{k+1}=V(\rho_k,\tau)U_k$ uses at most $b$
calls to $U_k$ and $U_k^\dagger$ within $V(\rho_k,\tau)$ and
one additional call to $U_k$.
Thus, $Q_{k+1}\le(1+b)Q_k$, which, together with $Q_0=1$,
gives Eq.~\eqref{eq:query-complexity-stability}.
The stated choices of $b$ follow from the circuit counts in
Secs.~\ref{sec:proposed-framework} and~\ref{subsec:one-step-error-NLSE}.
\end{proof}

\noindent
Under the hypotheses of Theorem~\ref{thm:query-complexity-Lipschitz-bounds},
assume $C_{\rm step}=O(1)$, $T\ge1$, and $0<\epsilon\le1$,
with fixed $\lambda>0$ in the exponential cases.
Simplifying the number of time steps in
Eq.~\eqref{eq:step-complexity-stability} under these assumptions
and substituting the resulting expressions into
Eq.~\eqref{eq:query-complexity-stability}
yields the following corollary.

\begin{corollary}
\label{cor:query-complexity-simplified}
Under the above assumptions, with $M$ chosen as in
Eq.~\eqref{eq:step-complexity-stability}, the query complexity satisfies
\begin{equation}
    Q_M \le
    \begin{dcases}
        \exp\left(
            O\left(
                \frac{T}{\epsilon}\log(1+b)
            \right)
        \right)
        & (\Lambda(t)=e^{-\lambda t}),\\
        \exp\left(
            O\left(
                \frac{T^2}{\epsilon}\log(1+b)
            \right)
        \right)
        & (\Lambda(t)=1),\\
        \exp\left(
            O\left(
                \frac{Te^{\lambda T}}{\epsilon}\log(1+b)
            \right)
        \right)
        & (\Lambda(t)=e^{\lambda t}).
    \end{dcases}
\end{equation}
\end{corollary}

To compare the dependence on $T$ alone, we fix the accuracy
and all model-dependent constants, including the positive
exponential rates, and assume $b=O(1)$ uniformly at the
required step sizes and accuracies.
Applying Corollary~\ref{cor:query-complexity-simplified}
with the exponential expansion bound $\Lambda(t)=e^{\lambda_{\rm exp}t}$ from
Lemma~\ref{lem:generic-exponential-solution-bound},
where $\lambda_{\rm exp}=2(M_G+L_G)$, 
yields an initial state preparation oracle query upper bound
of $\exp(\exp(O(T)))$ in the general case.
The same doubly exponential dependence on $T$ in the query upper bound 
also follows by 
applying the DB-QITE constructions and error-accumulation analyses developed in previous work~\cite{Gluza2026DBQITE,Wright_2026_DBTFD}
to the present state-dependent problem under our assumptions.

By introducing the trace-norm Lipschitz bound $\Lambda(t)$
in Definition~\ref{def:Lipschitz-bound} into the global error
and query analysis, we can improve this general upper bound
when tighter bounds are available for the target nonlinear
dynamics.
In particular, if the dynamics admits either
the exponential contraction bound $\Lambda(t)=e^{-\lambda t}$
with $\lambda>0$ or the nonexpansive bound $\Lambda(t)=1$,
Corollary~\ref{cor:query-complexity-simplified}
gives initial state preparation oracle query upper bounds
of $\exp(O(T))$ or $\exp(O(T^2))$, respectively.
Thus, our analysis classifies the initial state
preparation oracle query upper bounds according to
$\Lambda(t)$ and identifies conditions on the target
nonlinear dynamics under which the general doubly
exponential bound can be improved.
As an example, in Appendix~\ref{appendix:ITE-contraction},
we establish an exponential contraction bound on a
specified subset $\Omega$ for the normalized imaginary-time
evolution targeted by DB-QITE~\cite{Gluza2026DBQITE},
assuming a time-independent Hamiltonian $H$
with a unique ground state.
There, for $\rho_0\in\Omega$ and provided that the states
generated by the algorithm remain in $\Omega$ and the
remaining conditions stated in the appendix hold,
we apply Corollary~\ref{cor:query-complexity-simplified}
to obtain an initial state preparation oracle query
upper bound $\exp(O(T))$ at fixed accuracy and
model-dependent constants.
This improves on our general doubly exponential bound
$\exp(\exp(O(T)))$.

We have derived error bounds and query complexity upper bounds
for the proposed framework.
In the next section, we apply this framework to the discrete GP equation.
We briefly describe an ancilla-free implementation and then investigate
the initial state preparation oracle query complexity
of a single-qubit instance.

\section{Application to the GP Equation}
\label{sec:application}
To examine how these stability-dependent query upper bounds compare with the optimal worst-case query complexity,
we apply the framework of Secs.~\ref{sec:proposed-framework}--\ref{sec:global-analysis}
to the discrete GP equation 
under suitable oracle-access assumptions.
We first summarize an ancilla-free implementation on a grid of
$N=2^n$ points.
The number of queries to the initial state preparation oracle is
bounded by a polynomial in $n$ at fixed evolution time and accuracy,
with model parameters bounded independently of $n$.
We then analyze the query complexity of a single-qubit instance
for a specified initial state family.
For this family, with fixed nonlinearity strength $g>0$ and
trace-norm accuracy $0<\epsilon<1$, we derive a worst-case query
lower bound of $\Omega(e^{gT/2})$ from the amplification of
distinguishability between initially close states~\cite{an_quantum_2025}.
Using trajectory-dependent error estimates, we then derive a
worst-case query upper bound of $\exp(O(T^2))$ for recursive
simulation along the solution trajectories.
This improves on the general doubly exponential bound but
does not coincide with the lower bound.
To determine whether the lower bound can be attained, we further
construct a separate ancilla-free recursive circuit tailored to
the exact solutions and relate its construction to the double-bracket
framework.
Its worst-case query upper bound is $O(e^{gT/2})$, 
achieving the optimal query complexity up to a constant factor at fixed accuracy.

\subsection{An Ancilla-Free Implementation of the Discrete GP Equation}
In this subsection, we outline an ancilla-free implementation of
a discrete GP equation whose initial state preparation oracle
query complexity is polylogarithmic in the number of grid points $N$
under the assumptions stated below.
On a grid with $N=2^n$ points, 
we consider a discrete GP equation
of the form
\begin{equation}
    i\frac{d}{dt}\ket{\psi(t)}
    =
    \left(K+g\cD(\rho(t))\right)\ket{\psi(t)},
    \label{eq:discrete-GP}
\end{equation}
where $K$ is a state-independent Hamiltonian,
$g\in\mathbb{R}$ is the nonlinearity strength,
$\rho(t)=\proj{\psi(t)}$, and
\begin{equation}
    \cD(\rho)
    =
    \sum_{j=0}^{N-1}
    \bra{j}\rho\ket{j}\ket{j}\bra{j}
\end{equation}
is the dephasing map.

We summarize below the main steps of the implementation
and the query-complexity analysis; full details are provided
in Appendix~\ref{appendix:discrete-GP}.
To apply the NLSE-specific circuit of
Sec.~\ref{subsec:one-step-error-NLSE} to Eq.~\eqref{eq:discrete-GP},
we need to approximate the exponentials of
$H(\rho)=K+g\cD(\rho)$.
Assuming ancilla-free access to $e^{iK\theta}$ for arbitrary
$\theta\in\mathbb{R}$, the first-order Lie--Trotter formula
reduces this task to approximating
$e^{\pm ig\cD(\rho)\tau/2}$.
Directly applying this formula to the full dephasing sum
over $N=2^n$ conjugations by tensor products of $\id$ and $Z$
uses $O(N)$ calls in total to the current state
preparation circuit and its inverse.
The resulting recursive construction therefore yields an initial
state preparation oracle query upper bound that is exponential
in $n$ even for a fixed number of time steps, and hence does not
establish a polylogarithmic dependence on $N =2^n$.
To avoid this exponential dependence on $n$, we approximate
the dephasing map using a small-bias set~\cite{ta-shma_explicit_2017}.
At fixed step size $\tau$ and model parameters, the resulting
one-step implementation requires $O(n)$ calls in total to the
current state preparation circuit and its inverse.
For fixed $T$ and $\epsilon$, with $\lvert g\rvert$ and
$\norm{K}_{\rm op}$ bounded independently of $n$,
the resulting recursive construction yields an initial
state preparation oracle query upper bound that is
polynomial in $n$, i.e., polylogarithmic in $N$.

\subsection{Problem Setup for the Query Complexity Analysis of Single-Qubit GP Simulation}
To examine the relation between nonlinear amplification of
distinguishability and simulation query complexity in an
analytically tractable setting, we now consider a single-qubit instance.
The initial states are specified by a known parameter
$0<\xi<1$, and the target evolution time $T>0$ is given
independently of $\xi$.
For each fixed pair $(\xi,T)$, one circuit must simulate
either of two specified initial states, with the choice
supplied only through the initial state preparation oracle.
We evaluate the worst-case query complexity over $\xi$
at each target evolution time $T$.

Choosing $K=gX/4$, where $g>0$ is known and $X$ and $Z$
denote the Pauli $X$ and $Z$ operators, respectively,
and using $\cD(\rho)=(\rho+Z\rho Z)/2$, we obtain
\begin{equation}
    i\frac{d}{dt}\ket{\psi(t)}
    =
    \left[
        \frac{g}{4}X
        +\frac{g}{2}\left(\rho(t)+Z\rho(t)Z\right)
    \right]\ket{\psi(t)}.
    \label{eq:single-qubit-GP}
\end{equation}
Writing $\rho=(\id+xX+yY+zZ)/2$, where $r=(x,y,z)$ is its
Bloch vector and $Y$ is the Pauli $Y$ operator, gives
\begin{align}
    \frac{d}{dt}x&=-gyz, \\
    \frac{d}{dt}y&=gz\left(x-\frac{1}{2}\right),\\
    \frac{d}{dt}z&=\frac{gy}{2}.
\end{align}
For each $0<\xi<1$, the two allowed initial states are
represented by the Bloch vectors
\begin{equation}
    r_\pm(0)=
    \left(1-\xi,\,
    \pm\sqrt{\xi(1-\xi)},\,
    \pm\sqrt{\xi}\right),
    \label{eq:single-qubit-GP-initial-bloch}
\end{equation}
with the same sign in the last two components.
In the Bloch representation, the corresponding exact solutions are
\begin{equation}
    r_\pm(t)=
    \left(
        \tanh^2 a(t),\,
        \pm\frac{\tanh a(t)}{\operatorname{cosh}a(t)},\,
        \pm\frac{1}{\operatorname{cosh}a(t)}
    \right),
    \label{eq:single-qubit-GP-solution-bloch}
\end{equation}
for all $t\ge0$, 
where 
$a(t)=a_0-gt/2$ with $a_0=\operatorname{artanh}\sqrt{1-\xi}$.
$\ket{\psi_\pm(t)}$ denote the normalized state vectors corresponding to
$r_\pm(t)$.
Their density operators are given by
$\rho_\pm(t)=\proj{\psi_\pm(t)}$.
The two initial states approach one another as $\xi\to0$.
Their solutions become orthogonal at $t=2a_0/g$,
and both converge as $t\to\infty$ to the stationary state
\begin{equation}
    \rho_{\rm sta}
    \coloneqq\proj{+},
\end{equation}
where $\ket{+}=(\ket{0}+\ket{1})/\sqrt2$.
We use their separation to derive a query lower bound
and their convergence to choose a stopping time
for the GP circuit.

The parameters $g,\xi,T$, and the desired trace-norm accuracy
$0 < \epsilon < 1$ are known when the circuit is designed.
The given unitary oracle $U_0$ prepares one of the two
states in Eq.~\eqref{eq:single-qubit-GP-initial-bloch},
but which candidate was prepared is not supplied as classical information.
The algorithm must produce a normalized output state
approximating the corresponding $\rho_\pm(T)$ within
trace-norm error $\epsilon$ 
for every unitary oracle $U_0$ satisfying
$U_0\ket{0}=\ket{\psi_+(0)}$ or
$U_0\ket{0}=\ket{\psi_-(0)}$.
For fixed $g,\xi,T,\epsilon$, the non-query gates and their
parameters, as well as the number and placement of oracle
queries, are identical for both candidate initial states
and independent of the unspecified action of $U_0$.
They may depend on the known parameters, so a different
circuit may be designed for a different value of $\xi$.

For an algorithm $\mathcal A$ satisfying these requirements,
let $Q_{\mathcal A}(g,T,\xi,\epsilon)$ be its query count,
maximized over all unitary oracles $U_0$ satisfying
$U_0\ket{0}=\ket{\psi_+(0)}$ or
$U_0\ket{0}=\ket{\psi_-(0)}$, and define
\begin{equation}
    Q_{\mathcal A}^{\rm wc}(g,T,\epsilon)
    \coloneqq\sup_{0<\xi<1}
       Q_{\mathcal A}(g,T,\xi,\epsilon).
\end{equation}
All upper bounds uniform in $\xi$ refer to this quantity.

\subsection{Worst-Case Query Lower Bound}
\label{subsec:single-qubit-GP-lower}
We establish a worst-case initial state preparation oracle
query lower bound for the single-qubit GP simulation problem
defined above.
To derive this bound, we use the following specialization
of Theorem~4 in Ref.~\cite{an_quantum_2025} to deterministic
algorithms with access to $U_0$ and $U_0^\dagger$.
\begin{theorem}[Theorem 4 in Ref.~\cite{an_quantum_2025}]
\label{thm:state-amplifier-query-lower-bound}
Let $\ket{\psi}$ and $\ket{\phi}$ be normalized states with
real overlap $\braket{\psi|\phi}\ge1-\eta$, where $0<\eta<1$.
Suppose a deterministic algorithm receives a unitary oracle
$U_0$ satisfying
$U_0\ket{0}=\ket{\psi}$ or $U_0\ket{0}=\ket{\phi}$.
For the fixed state pair and fixed known classical parameters,
the non-query gates and their parameters are the same for
all such oracles, as are the number and placement of queries
and whether each query applies $U_0$ or $U_0^\dagger$.
For every pair of oracles preparing the respective states,
suppose the corresponding output density operators are
separated by $\Omega(1)$ in trace norm, with a constant
independent of $\eta$ and the choice of oracles.
Then the algorithm requires $\Omega(\eta^{-1/2})$ queries
to $U_0$ or $U_0^\dagger$ in the worst case over all such
unitary oracles.
\end{theorem}

Fix a target evolution time $T>0$ and set
\begin{equation}
    \eta=\frac{1}{\operatorname{cosh}^2(gT/2)}.
\end{equation}
Consider the pair in
Eq.~\eqref{eq:single-qubit-GP-initial-bloch} with $\xi=\eta$,
choosing the phases of the initial state vectors so that
their overlap is real and nonnegative.
Then
\begin{equation}
    \braket{\psi_+(0)|\psi_-(0)}=1-\eta.
\end{equation}
This choice gives $a_0=gT/2$ and hence
$\rho_+(T)=\proj{0}$ and $\rho_-(T)=\proj{1}$.
For a fixed accuracy $0<\epsilon<1$, output density operators $\widetilde\rho_\pm(T)$
satisfying
\begin{equation}
    \norm{\widetilde\rho_\pm(T)-\rho_\pm(T)}_1
    \le\epsilon
\end{equation}
obey
$\norm{\widetilde\rho_+(T)-\widetilde\rho_-(T)}_1
\ge2-2\epsilon>0$.
Theorem~\ref{thm:state-amplifier-query-lower-bound}
therefore gives
\begin{equation}
    Q_{\mathcal A}^{\rm wc}(g,T,\epsilon)
    \ge Q_{\mathcal A}(g,T,\eta,\epsilon)
    =\Omega\!\left(\cosh\frac{gT}{2}\right)
    =\Omega\!\left(e^{gT/2}\right).
    \label{eq:single-qubit-GP-query-lower-bound}
\end{equation}
The choice $\xi=\eta$ selects instances for this worst-case
lower bound without restricting the independent inputs
$\xi$ and $T$ of the simulation problem.

\subsection{Trajectory Simulation and Query Upper Bound}
\label{subsec:single-qubit-GP-error-upper}
We apply the NLSE-specific circuit of
Sec.~\ref{subsec:one-step-error-NLSE} to
Eq.~\eqref{eq:single-qubit-GP} to approximate the solution
trajectory starting from the given initial state.
We derive an initial state preparation oracle query upper
bound uniform over $0<\xi<1$ for every independently
specified target evolution time $T>0$.

For this single-qubit instance, the dephasing map has the exact
representation $\cD(\rho)=(\rho+Z\rho Z)/2$.
Applying the first-order Lie--Trotter formula, we define an
ancilla-free approximation to
$e^{\pm i(K+g\cD(\rho))\tau/2}$ by
\begin{equation}
    A_H(\rho,\pm\tau/2)
    =e^{\pm i(g\tau/8)X}
     e^{\pm i(g\tau/4)\rho}
     Z e^{\pm i(g\tau/4)\rho}Z.
    \label{eq:single-qubit-GP-H-step}
\end{equation}
Using this implementation in Eq.~\eqref{eq:NLSE-one-step},
define
\begin{equation}
    \widetilde\Phi_\tau^{({\rm GP})}(\rho)
    \coloneqq
    V_{\rm NLSE}(\rho,\tau)\rho
    V_{\rm NLSE}(\rho,\tau)^\dagger.
    \label{eq:single-qubit-GP-implemented-map}
\end{equation}

Lemma~\ref{lem:single-qubit-GP-relative-difference} gives 
an exponential expansion bound
and bounds the one-step error in terms of the distance
from $\rho_{\rm sta}$.
\begin{lemma}
\label{lem:single-qubit-GP-relative-difference}
There exist constants $C,\Delta_0>0$, independent of $g$, $T$,
and $\xi$, such that, for every $\tau$ satisfying
$0<g\tau\le\Delta_0$ and all pure single-qubit states
$\rho,\sigma$,
\begin{align}
    \norm{\varphi_\tau(\rho)-\varphi_\tau(\sigma)}_1
    &\le e^{g\tau/2}\norm{\rho-\sigma}_1,
    \label{eq:single-qubit-GP-solution-bound}\\
    \norm{\widetilde\Phi_\tau^{({\rm GP})}(\rho)
        -\varphi_\tau(\rho)}_1
    &\le C(g\tau)^2\norm{\rho-\rho_{\rm sta}}_1.
    \label{eq:single-qubit-GP-relative-difference}
\end{align}
The first inequality holds for all $\tau\ge0$.
\end{lemma}
\begin{proof}
See Appendix~\ref{appendix:proof-single-qubit-GP-relative-difference}.
\end{proof}

\noindent
Using $\norm{\rho-\rho_{\rm sta}}_1\le2$ in
Eq.~\eqref{eq:single-qubit-GP-relative-difference} gives a uniform
one-step error bound of $2C(g\tau)^2$.
For $0<\tau\le\min\{1,\Delta_0/g\}$, we apply the global
error analysis of Sec.~\ref{sec:global-analysis} with 
the exponential expansion bound
$\Lambda(t)=e^{gt/2}$ and $b=O(1)$.
At fixed $g>0$ and accuracy, this gives an initial state
preparation oracle query upper bound of
$\exp(\exp(O(T)))$ for sufficiently large $T$,
as in Corollary~\ref{cor:query-complexity-simplified}.

We obtain a tighter bound by retaining the dependence of
the one-step error on $\norm{\rho-\rho_{\rm sta}}_1$ and
using the exact solution trajectories to estimate the
accumulated error and choose the stopping time.

\begin{theorem}[Query upper bound from trajectory-dependent error analysis]
\label{thm:single-qubit-GP-query-upper-bound}
Assume access to an initial state preparation oracle $U_0$
satisfying
$U_0\ket{0}=\ket{\psi_+(0)}$ or
$U_0\ket{0}=\ket{\psi_-(0)}$
and its inverse $U_0^\dagger$. 
Then there exists a deterministic, fully coherent,
ancilla-free algorithm $\mathcal A$ based on the
implemented one-step map in
Eq.~\eqref{eq:single-qubit-GP-implemented-map}
that satisfies the simulation requirements stated above
for all $g>0$, $T>0$, $0<\xi<1$, and $0<\epsilon<1$.
Its worst-case initial state preparation oracle query
count satisfies
\begin{equation}
    Q_{\mathcal A}^{\rm wc}(g,T,\epsilon)
    \le\exp\left[O\left(
        1+gT+\frac{(gT)^2}{\epsilon^2}
    \right)\right].
    \label{eq:single-qubit-GP-query-upper-bound}
\end{equation}
The implied constant is independent of $g,T,\xi$, and
$\epsilon$.
\end{theorem}

The proof proceeds as follows.
After both exact solutions have approached
$\rho_{\rm sta}$ sufficiently closely, the state prepared
at a cutoff time can approximate the solution at later
times within the desired accuracy.
We choose this cutoff to control both the error from
stopping the simulation and the accumulated one-step errors.
Fix an accuracy $0<\epsilon<1$ and define the stopping time by
\begin{equation}
    S=\min\{T,t_{\rm cut}\},
    \label{eq:single-qubit-GP-simulation-duration}
\end{equation}
where the cutoff time is given by
\begin{equation}
    t_{\rm cut}
    =\frac{2}{g}\left(
        a_0+\log\frac{8\sqrt2}{\epsilon}
      \right).
\end{equation}
The cutoff time $t_{\rm cut}$ depends only on the known
parameters and is later than $2a_0/g$, the time at which
the two exact solutions become orthogonal.
Equation~\eqref{eq:single-qubit-GP-solution-bloch} gives
\begin{equation}
    \norm{\rho_\pm(t)-\rho_{\rm sta}}_1
    =\frac{\sqrt{2}}{\operatorname{cosh}a(t)}.
\end{equation}
For $t\ge t_{\rm cut}$, using
$1/\operatorname{cosh}a\le2e^a$ yields
\begin{equation}
    \norm{\rho_\pm(t)-\rho_{\rm sta}}_1
    \le\frac{\epsilon}{4}.
\end{equation}
If $S<T$, the triangle inequality therefore implies
\begin{equation}
    \norm{\rho_\pm(S)-\rho_\pm(t)}_1
    \le\frac{\epsilon}{2}
    \qquad \left( S\le t\le T \right).
    \label{eq:single-qubit-GP-duration-error}
\end{equation}
The same inequality holds trivially when $S=T$.
Thus, it suffices to simulate up to $S$ with error at most
$\epsilon/2$ and use the resulting state for later times.

For the recursion up to $S$, choose a step size $\tau>0$
such that $M \coloneqq S/\tau$ is a positive integer,
and set $t_k=k\tau$.
Start the recursion with the given oracle $U_0$.
For either candidate initial state, define
\begin{equation}
    \rho_{k+1}^\pm
    =\widetilde\Phi_\tau^{({\rm GP})}(\rho_k^\pm)
    \quad \left(k \in \{0, 1, \dots, M-1 \} \right)
\end{equation}
with $\rho_0^\pm=\rho_\pm(0)$.
The recursive circuit
$U_{k+1}=V_{\rm NLSE}(\rho_k^\pm,\tau)U_k$ prepares
$\rho_{k+1}^\pm$ for the respective given oracle.
For fixed $g,\xi,T,\epsilon$, the stopping time $S$,
the number of time steps $M$, the non-query gates and their
parameters, and the number and placement of queries are
the same for both candidate initial states.
Whether each query applies $U_0$ or $U_0^\dagger$ is also the
same, and these choices are independent of the unspecified
action of $U_0$.

Set $e_k^\pm=\norm{\rho_k^\pm-\rho_\pm(t_k)}_1$.
For $0<g\tau\le\Delta_0$, Lemma~\ref{lem:single-qubit-GP-relative-difference} and the triangle
inequality give
\begin{equation}
    \begin{split}
        e_{k+1}^\pm
        \le{}&\left(e^{g\tau/2}+C(g\tau)^2\right)e_k^\pm\\
        &+C(g\tau)^2
          \norm{\rho_\pm(t_k)-\rho_{\rm sta}}_1.
    \end{split}
\end{equation}
For all real $a$, 
$1/\operatorname{cosh}a\le2e^{-a}$.
Consequently, along the simulated trajectory,
\begin{equation}
    \norm{\rho_\pm(t_k)-\rho_{\rm sta}}_1
    \le2\sqrt2e^{-a_0+gt_k/2}.
\end{equation}
Since $e_0^\pm=0$ and
$e^{g\tau/2}+C(g\tau)^2\le
e^{g\tau/2+C(g\tau)^2}$, expanding the recurrence yields,
for $k \in \{1,2,\dots,M\}$,
\begin{align}
    e_k^\pm
    &\le2\sqrt2 C(g\tau)^2
      e^{-a_0+g(t_k-\tau)/2}
      \sum_{j=0}^{k-1}e^{C(k-1-j)(g\tau)^2}\\
    &\le2\sqrt2 C e^{-a(S)}
      \frac{(gS)^2}{M}
      \exp\left(C\frac{(gS)^2}{M}\right).
\end{align}
The trajectory-dependent factor in the one-step error
compensates for the subsequent exponential error
amplification, while the stopping time controls the
remaining factor $e^{-a(S)}$.
Indeed, $S\le t_{\rm cut}$ gives
$e^{-a(S)}\le8\sqrt2/\epsilon$.
Hence
\begin{equation}
    \max_{0\le k\le M}e_k^\pm
    \le\frac{32C}{\epsilon}
       \frac{(gS)^2}{M}
       \exp\left(C\frac{(gS)^2}{M}\right).
    \label{eq:single-qubit-GP-global-error}
\end{equation}

As shown in Appendix~\ref{appendix:single-qubit-GP-step-count},
a sufficient number of time steps is
\begin{equation}
    M=O\left(1+gS+\frac{(gS)^2}{\epsilon^2}\right),
    \label{eq:single-qubit-GP-step-count}
\end{equation}
which ensures $g\tau\le\Delta_0$ and
\begin{equation}
    \max_{0\le k\le M}
    \norm{\rho_k^\pm-\rho_\pm(t_k)}_1
    \le\frac{\epsilon}{2}.
\end{equation}
Combining the case $k=M$ with
Eq.~\eqref{eq:single-qubit-GP-duration-error} gives
\begin{equation}
    \norm{\rho_M^\pm-\rho_\pm(t)}_1\le\epsilon
    \qquad (S\le t\le T).
\end{equation}
Thus, $U_k$ prepares an approximation to the original solution
at each grid point $t_k\le S$, and $U_M$ prepares an
approximation at every time in $[S,T]$.
In particular, its output satisfies the prescribed accuracy
at the target evolution time $T$.

Recall from Sec.~\ref{subsec:general-framework} that
$Q_k$ counts calls to $U_0$ and $U_0^\dagger$ in $U_k$.
Each $A_H$ factor in Eq.~\eqref{eq:single-qubit-GP-H-step}
contains two exponentials of the projector onto the current state,
each implemented using one call to $U_k$ and one to
$U_k^\dagger$.
Thus, the two $A_H$ factors in
Eq.~\eqref{eq:NLSE-one-step} require eight calls in total
to $U_k$ and $U_k^\dagger$, and the two current state
reflections require four additional calls.
Including the rightmost $U_k$ in the recursive update gives
$Q_{k+1}\le13Q_k$.
Since $Q_0=1$, we obtain $Q_M\le13^M$.
Combining $Q_M\le13^M$ with
Eq.~\eqref{eq:single-qubit-GP-step-count} and $S\le T$,
and taking the supremum over $0<\xi<1$, yields
Eq.~\eqref{eq:single-qubit-GP-query-upper-bound}.
The implied constant is independent of $g,T,\xi$, and
$\epsilon$.
The construction applies to every unitary oracle satisfying
either initial state preparation condition.
All these circuits are deterministic, fully coherent,
and ancilla-free.
This proves Theorem~\ref{thm:single-qubit-GP-query-upper-bound}.

At fixed $g>0$ and accuracy, trajectory-dependent error
estimates and a stopping time improve the initial state
preparation oracle query upper bound from
$\exp(\exp(O(T)))$ to the singly exponential bound
$\exp(O(T^2))$, uniformly over the allowed initial state family.
This demonstrates the usefulness of the proposed double-bracket
framework: it propagates the solution from one time step
to the next while achieving a query upper bound substantially
closer to the exponential worst-case lower bound of
Sec.~\ref{subsec:single-qubit-GP-lower}.
In Sec.~\ref{subsec:single-qubit-GP-achieving-optimal},
we construct a separate recursive circuit tailored to the
exact solutions and relate its reflection updates to the
double-bracket framework.
Its query upper bound coincides with the lower bound up to a
constant factor at fixed trace-norm accuracy.

\subsection{Achieving Optimal Worst-Case Query Complexity without Following the Solution Trajectories}
\label{subsec:single-qubit-GP-achieving-optimal}
The GP circuit analyzed in the preceding subsection propagates the solution from one time step to the next. 
The simulation task itself, however, requires only the state at the target evolution time $T$, so the intermediate states of the recursion need not follow the solution trajectories. 
We exploit this freedom to construct a separate recursive circuit tailored to the exact solutions.
The recursion uses products of reflections that exactly
implement a one-step propagator of the form introduced in
Sec.~\ref{subsec:general-framework}.
Its initial state preparation oracle query upper bound is
uniform in $\xi$ and coincides with the worst-case lower bound of
Sec.~\ref{subsec:single-qubit-GP-lower} up to a constant factor
at fixed trace-norm accuracy.
The construction applies for all $0<\xi<1$ and $T>0$,
with $\xi$ and $T$ specified independently.

To construct the circuit, define
$\alpha(t)=\arcsin\left(  1/\left( \sqrt2 \operatorname{cosh}a(t) \right) \right)$ 
and $\beta(t)=\arctan(\tanh a(t))$.
With $\ket{-}=(\ket{0}-\ket{1})/\sqrt2$, the exact
solutions satisfy, up to a global phase,
\begin{equation}
    \ket{\psi_\pm(t)}
    =\cos\alpha(t)\ket{+}
      \pm e^{-i\beta(t)}\sin\alpha(t)\ket{-}.
\end{equation}
For a real angle $u$, write $\widehat\rho_\pm(u)=\proj{\chi_\pm(u)}$, 
where $\ket{\chi_\pm(u)}=\cos u\ket{+}\pm\sin u\ket{-}$.
The same known rotation $e^{-i\beta(t)X/2}$ removes the
relative phase for both candidates and maps
$\ket{\psi_\pm(t)}$ to $\ket{\chi_\pm(\alpha(t))}$ up to
a global phase.
Thus, $u=\alpha(t)$ represents the solution at time $t$
after this rotation, with $0<\alpha(t)\le\pi/4$.

The remaining task is therefore to change $\alpha(0)$
to $\alpha(T)$ for both candidates using the same circuit,
and then restore the phase specified by $\beta(T)$.
Theorem~\ref{thm:single-qubit-GP-optimal-query-bound} below shows that this task admits a
deterministic, fully coherent, ancilla-free recursive circuit
with optimal worst-case initial state preparation oracle
query complexity at fixed trace-norm accuracy.

\begin{theorem}[Optimal worst-case query complexity via a circuit tailored to exact solutions]
\label{thm:single-qubit-GP-optimal-query-bound}
Under the oracle assumptions stated in
Theorem~\ref{thm:single-qubit-GP-query-upper-bound},
there exists a deterministic, fully coherent, ancilla-free
algorithm $\mathcal A$ implemented by a recursive circuit
tailored to the exact solutions that satisfies the simulation
requirements stated above for all $g>0$, $T>0$,
$0<\xi<1$, and $0<\epsilon<1$.
It prepares 
$\rho_\pm(T)$ exactly, and its worst-case initial state
preparation oracle query count satisfies
\begin{equation}
    Q_{\mathcal A}^{\rm wc}(g,T,\epsilon)
    \le\frac{3\pi}{2\sqrt2}e^{gT/2}
    \le\frac{3\pi}{\sqrt2}\cosh\frac{gT}{2}.
    \label{eq:single-qubit-GP-tighter-upper-bound}
\end{equation}
The optimal worst-case initial state preparation oracle
query complexity, denoted $Q_{\rm opt}^{\rm wc}(g,T,\epsilon)$,
satisfies
\begin{equation}
    Q_{\rm opt}^{\rm wc}(g,T,\epsilon)
    =\Theta\left(\cosh\frac{gT}{2}\right)
    =\Theta\left(e^{gT/2}\right)
    \label{eq:single-qubit-GP-optimal-query-bound}
\end{equation}
for every fixed $0<\epsilon<1$ and sufficiently large $gT$.
The worst case is taken over the specified initial state
family and the allowed preparation oracles, with $\xi$
supplied as classical information.
\end{theorem}

The proof proceeds as follows.
We construct the recursive update using the reflections
$R_\rho=\id-2\rho$ and $R_{\rho_{\rm sta}}=-X$.
For $0<u\le\pi/4$, direct multiplication in the basis
$\{\ket{+},\ket{-}\}$ gives
\begin{equation}
    R_{\widehat\rho_\pm(u)}R_{\rho_{\rm sta}}
    =\begin{pmatrix}
       \cos2u&\mp\sin2u\\
       \pm\sin2u&\cos2u
     \end{pmatrix}.
\end{equation}
Equivalently,
\begin{equation}
    R_{\widehat\rho_\pm(u)}R_{\rho_{\rm sta}}
    =
    \exp\left(
        \frac{4u}{\sin2u}
        \comm{\widehat\rho_\pm(u)}{\rho_{\rm sta}}
    \right).
    \label{eq:single-qubit-GP-reflection-commutator}
\end{equation}
This is the one-step propagator in
Eq.~\eqref{eq:discretized-one-step-propagator} for
$G(\rho)=\rho_{\rm sta}$, with the commutator held fixed at
the input state and step parameter $4u/\sin2u$.
Its action triples the angle:
\begin{equation}
    R_{\widehat\rho_\pm(u)}R_{\rho_{\rm sta}}
    \ket{\chi_\pm(u)}
    =\ket{\chi_\pm(3u)}.
\end{equation}
The target GP state at time $T$ is specified by
$\alpha(T)$ and $\beta(T)$.

The angle $\alpha(t)$ increases until $t=2a_0/g$ and
decreases thereafter, so $\alpha(T)$ may be smaller than
$\alpha(0)$.
The final adjustment described below allows both an increase
and a decrease of the angle.
Set $\alpha_0=\alpha(0)$, choose
\begin{equation}
    m=\max\left\{0,
       \left\lfloor\log_3\frac{\alpha(T)}{\alpha_0}\right\rfloor
       \right\},
    \label{eq:single-qubit-GP-tripling-count}
\end{equation}
and define $\alpha_k=3^k\alpha_0$ for
$k\in\{0,1,\dots,m\}$.
Starting from $\widehat U_0=e^{-i\beta(0)X/2}U_0$, define
\begin{equation}
    \widehat U_{k+1}
    =\widehat U_k(\id-2\proj{0})\widehat U_k^\dagger
       R_{\rho_{\rm sta}}\widehat U_k,
\end{equation}
for $k \in \{0,1, \dots, m-1 \}$.
Then $\widehat U_k\ket{0}=\ket{\chi_\pm(\alpha_k)}$
up to a global phase for either candidate, and
$0<\alpha_k\le\pi/4$ throughout the recursion.
Each recursive step uses three calls in total to
$\widehat U_k$ and $\widehat U_k^\dagger$.
Counting calls to $U_0$ and $U_0^\dagger$ together gives
$Q_{k+1}=3Q_k$ and $Q_0=1$, and hence $Q_k=3^k$.

It remains to change the angle from $\alpha_m$ to $\alpha(T)$.
The choice of $m$ ensures
$0<\alpha_m\le\pi/4$ and
$0<\alpha(T)\le\min\{3\alpha_m,\pi/4\}$.
The target of this adjustment is
$\ket{\chi_\pm(\alpha(T))}
=\cos\alpha(T)\ket{+}\pm\sin\alpha(T)\ket{-}$.
We choose $\theta$ so that the projector exponentials below
give coefficient magnitudes $\cos\alpha(T)$ and $\sin\alpha(T)$
in the basis $\{\ket{+},\ket{-}\}$, respectively.
The unitary $D_{\alpha_m,\alpha(T)}$ then corrects the resulting
phases to obtain this state.
The correction $D_{\alpha_m,\alpha(T)}$ is diagonal in
the basis $\{\ket{+},\ket{-}\}$ and is therefore an
$X$ rotation up to a global phase.
Explicit choices of $\theta$ and $D_{\alpha_m,\alpha(T)}$
are given in Appendix~\ref{appendix:single-qubit-GP-final-adjustment}.

Since $\widehat U_m\ket{0}=\ket{\chi_\pm(\alpha_m)}$
up to a global phase, the circuit
\begin{equation}
    \begin{split}
        \mathcal P_{\alpha_m,\alpha(T)}[\widehat U_m]
        \coloneqq{}&
        D_{\alpha_m,\alpha(T)}
        \widehat U_m e^{i\theta\proj{0}}\widehat U_m^\dagger e^{i\theta\rho_{\rm sta}}\widehat U_m
    \end{split}
    \label{eq:single-qubit-GP-angle-change-circuit}
\end{equation}
prepares $\ket{\chi_\pm(\alpha(T))}$ up to a global phase.
This adjustment uses three calls in total to $\widehat U_m$
and $\widehat U_m^\dagger$.
All its non-query gates depend only on $\alpha_m$ and
$\alpha(T)$, and are the same for both candidates.

Restoring the target relative phase specified by $\beta(T)$
gives the final circuit
\begin{equation}
    U_{\rm f}
    =e^{i\beta(T)X/2}
      \mathcal P_{\alpha_m,\alpha(T)}[\widehat U_m].
\end{equation}
This prepares 
$\rho_\pm(T)$ exactly, 
including when $\alpha(T)<\alpha_0$.
All gate parameters and the number of recursive steps depend
only on the known $g,\xi,T$.
No information identifying the given candidate is needed.

As shown in Appendix~\ref{appendix:single-qubit-GP-final-adjustment},
the worst-case initial state preparation oracle query count
satisfies Eq.~\eqref{eq:single-qubit-GP-tighter-upper-bound}.
The circuit is deterministic, fully coherent, and ancilla-free,
and it works for every unitary oracle satisfying either
initial state preparation condition.

Combining Eq.~\eqref{eq:single-qubit-GP-tighter-upper-bound}
with Eq.~\eqref{eq:single-qubit-GP-query-lower-bound}
yields Eq.~\eqref{eq:single-qubit-GP-optimal-query-bound}.
This proves Theorem~\ref{thm:single-qubit-GP-optimal-query-bound}.
The exact representation of the reflection updates as
one-step propagators in
Eq.~\eqref{eq:single-qubit-GP-reflection-commutator}
connects this optimal circuit to the double-bracket framework.
The reflection update can also be viewed as a Grover-type amplitude-amplification step, consistent with the double-bracket interpretation of Grover's algorithm~\cite{Suzuki_2025_GroverITE}.

This optimality result complements previous
work~\cite{BrustleWiebe2025}, which uses Carleman linearization
to derive an exponential query upper bound in $T$ for a
broader class of norm-preserving polynomial differential
equations but does not establish that the upper and lower
bounds coincide up to a constant factor.
For the specified initial state family, our construction
achieves the optimal worst-case initial state
preparation oracle query complexity at fixed trace-norm accuracy,
even under the ancilla-free constraint.

\section{Conclusion and Discussion}
\label{sec:conclusion}
In this work, 
we generalize DB-QITE to nonlinear dynamics 
by replacing a state-independent Hamiltonian with a state-dependent Hermitian operator $G(\rho)$.
Under the boundedness and Lipschitz continuity assumptions
on $G(\rho)$ and 
the oracle-access and implementation assumptions
described in Secs.~\ref{sec:proposed-framework}
and~\ref{sec:one-step-error-bounds},
we develop a fully coherent, ancilla-free framework that constructs a recursive circuit from the initial state
preparation oracle and its inverse.
We establish one-step error bounds
for the general and NLSE-specific implementations.
Combining these one-step error bounds with a trace-norm Lipschitz bound
$\Lambda(t)$ yields global error bounds and a classification
of initial state preparation oracle query upper bounds.
At fixed model parameters and with a uniformly bounded
recursion factor, 
these query upper bounds are
$\exp(O(T/\epsilon))$ under the exponential contraction bound
$\Lambda(t)=e^{-\lambda t}$,
$\exp(O(T^2/\epsilon))$ under the nonexpansive bound
$\Lambda(t)=1$, and
$\exp(O(Te^{\lambda T}/\epsilon))$ under the exponential expansion bound
$\Lambda(t)=e^{\lambda t}$,
for $T\ge1$, $0<\epsilon\le1$, and fixed $\lambda>0$
in the exponential cases.
Under the same assumptions, the exponential expansion bound
$\Lambda(t)=e^{\lambda_{\rm exp}t}$ holds in the general case,
yielding an initial state preparation oracle query upper bound
of $\exp(\exp(O(T)))$ at fixed accuracy.
Adapting earlier DB-QITE constructions and
error-accumulation analyses~\cite{Gluza2026DBQITE,Wright_2026_DBTFD}
to the present state-dependent problem also gives a doubly exponential query upper bound in $T$.
Our analysis improves this general doubly
exponential upper bound under either the exponential contraction bound or the nonexpansive bound.

To examine how these stability-dependent query upper bounds compare with the optimal worst-case query complexity,
we apply the proposed framework to the discrete GP equation 
under the oracle-access and implementation assumptions stated in
Sec.~\ref{sec:application}.
For general $N=2^n$, the small-bias construction gives an
initial state preparation oracle query upper bound polynomial
in $n$ at fixed $T$, accuracy, and model parameters.
For the specified single-qubit initial state family with known
$\xi$ and independently specified $T$, 
we apply the general
result of Ref.~\cite{an_quantum_2025} to translate the
amplification of distinguishability between initially close
states under the nonlinear dynamics into a worst-case query
lower bound of $\Omega(e^{gT/2})$
at fixed
$0<\epsilon<1$, where $g>0$ is the nonlinearity strength.
Trajectory-dependent error estimates and a stopping time
improve the recursive GP circuit's query upper bound from
$\exp(\exp(O(T)))$ to $\exp(O(T^2))$ at fixed $g$ and
$\epsilon$, uniformly in $\xi$.
This demonstrates the usefulness of the proposed double-bracket
framework for propagating the solution from one time step
to the next with a singly exponential query upper bound,
although this bound does not coincide with the lower bound.
To determine whether the lower bound can be attained within
the fully coherent and ancilla-free formulation,
we construct a separate ancilla-free recursive circuit
tailored to the exact solutions.
The exact representation of its reflection updates as
double-bracket one-step propagators connects this
construction to the proposed framework.
The recursive circuit tailored to the exact solutions yields
a worst-case query upper bound that coincides with the lower bound
up to a constant factor, establishing the optimal worst-case
query complexity $\Theta(e^{gT/2})$ at fixed $0<\epsilon<1$
for sufficiently large $gT$.
For comparison, previous work~\cite{BrustleWiebe2025}, which
uses Carleman linearization for a broader class of
norm-preserving polynomial differential equations, derives
an exponential query upper bound in $T$ without establishing the optimality up to a constant factor.
Our construction thus attains optimal worst-case query complexity
for the specified initial state family at fixed trace-norm accuracy,
even under the ancilla-free constraint.

Taken together, our results establish a quantitative connection
between the sensitivity of nonlinear solution dynamics and
the initial state preparation oracle query complexity of their
fully coherent, ancilla-free embedding into unitary quantum mechanics.
Within the proposed framework, trace-norm stability controls the propagation
of simulation errors and leads to qualitatively different query
upper bounds.
These findings suggest that trace-norm stability can serve as a complexity parameter
for coherent nonlinear simulation beyond normalized imaginary-time
evolution.
These results, including the optimality established for the single-qubit GP instance,
also provide a foundation for designing ancilla-free unitary embeddings that achieve optimal query complexity.

A natural direction for future work is to identify further concrete examples
of nonlinear differential equations to which our framework applies,
and to examine to what extent trace-norm stability can be used to
characterize the initial state preparation oracle query complexity.
Another future direction is to extend the proposed framework
to an explicitly time- and state-dependent Hermitian operator $G(t,\rho)$.
This would allow us to investigate the relation between
the sensitivity of nonlinear solution dynamics and the initial
state preparation oracle query complexity for a broader class
of dynamics than that considered here.
A further question is when optimal query complexity
for simulating nonlinear dynamics
can be attained by fully coherent, ancilla-free constructions.

\section*{Statement of AI use}
In this work, the idea of generalizing the double-bracket
quantum algorithms by replacing the state-independent Hamiltonian with a
state-dependent Hermitian operator originated with the authors.
OpenAI models GPT-5.6-sol, GPT-5.6-luna, GPT-5.6-terra, GPT-5.6-sol-wm, and GPT-6-astra, accessed through ChatGPT in August and September
2026, were used to explore proof approaches, check intermediate
mathematical arguments, error estimates, and query-complexity bounds,
and assist with literature searches.
They were also used extensively to draft and revise the entire
manuscript, including the mathematical proofs.
The mathematical arguments and proofs were independently
verified by the authors.
The authors take sole responsibility for the correctness of the final manuscript.

\begin{acknowledgments}
This work is supported by MEXT Quantum Leap Flagship Program (MEXT Q-LEAP) Grant No. JPMXS0120319794, JST COI-NEXT Grant No. JPMJPF2014, and JST CREST JPMJCR24I3.
YI is supported by JST SPRING Grant Number JPMJSP2138.
HH is supported by JST PRESTO, Japan, Grant Number JPMJPR25F7 and JSPS KAKENHI Grant No. JP24K16979.
\end{acknowledgments}

\bibliography{references}

\appendix
\section{Purity Preservation under the State-Dependent Double-Bracket Flow}
\label{appendix:purity-preservation}
Let $\rho(t)$ be a differentiable Hermitian solution of
Eq.~\eqref{eq:state-dependent-db} on an interval containing $t=0$.
Since $G(\rho(t))$ is Hermitian,
$\comm{\rho(t)}{G(\rho(t))}$ is anti-Hermitian.
Let $U(t)$ satisfy
\begin{equation}
\frac{d}{dt}U(t)
=
\comm{\rho(t)}{G(\rho(t))}U(t)
\label{eq:purity-propagator}
\end{equation}
with $U(0)=\id$.
The anti-Hermiticity of the generator $\comm{\rho(t)}{G(\rho(t))}$ implies that $U(t)$ is unitary.
Using Eqs.~\eqref{eq:state-dependent-db} and~\eqref{eq:purity-propagator}, we obtain
\begin{align}
&\frac{d}{dt}
\left(
U(t)^\dagger\rho(t)U(t)
\right) \notag \\
=&
U(t)^\dagger
\Bigl(
-\comm{\rho(t)}{G(\rho(t))}\rho(t)
\notag\\
&\quad
+\comm{\comm{\rho(t)}{G(\rho(t))}}{\rho(t)}
+\rho(t)\comm{\rho(t)}{G(\rho(t))}
\Bigr)
U(t) \\
=&0.
\end{align}
Therefore,
$\rho(t)=U(t)\rho(0)U(t)^\dagger$, so the evolution preserves all eigenvalues of $\rho(t)$.
In particular, if $\rho(0)\in\mathcal{S}_{\rm pure}$, then $\rho(t)$
remains a rank-one projector and
$\operatorname{Tr}(\rho(t)^2)=1$ for all times for which the solution
exists.
Thus, $\rho(t)\in\mathcal{S}_{\rm pure}$ throughout the evolution.

\section{Proof of Lemma~\ref{lem:generic-exponential-solution-bound}} 
\label{appendix:Lem-generic-exponential-solution-bound}
We prove Lemma~\ref{lem:generic-exponential-solution-bound}.
\vspace{1em}
\newline
\textit{Proof of Lemma~\ref{lem:generic-exponential-solution-bound}}
\hspace{1em}
Define the double-bracket vector field by
\begin{equation}
    \cV(\rho)
    \coloneqq
    \comm{\comm{\rho}{G(\rho)}}{\rho}.
\end{equation}
For a rank-one projector $\rho$, let
\begin{equation}
    R_\rho\coloneqq\id-2\rho.
\end{equation}
Since $\rho^2=\rho$, the operator $R_\rho$ is a Hermitian unitary and
\begin{align}
    \cV(\rho)
    &=
    2\rho G(\rho)\rho-G(\rho)\rho-\rho G(\rho) \\
    &=
    \frac{1}{2}
    \left(
        R_\rho G(\rho)R_\rho-G(\rho)
    \right).
    \label{eq:DB-vector-field-reflection-form}
\end{align}
For any linear operator $X$ on $\cH$,
\begin{equation}
    \frac{1}{2}\left(R_\rho X R_\rho-X\right)
    =
    -\rho X(\id-\rho)-(\id-\rho)X\rho.
\end{equation}
Using $\norm{\rho}_1=1$ and
$\norm{\id-\rho}_{\rm op} \le 1$, we therefore obtain
\begin{align}
    \norm{
        \frac{1}{2}\left(R_\rho X R_\rho-X\right)
    }_1
    &\le
    \norm{\rho X(\id-\rho)}_1
    +
    \norm{(\id-\rho)X\rho}_1 \\
    &\le
    2\norm{X}_{\rm op}.
    \label{eq:rank-one-reflection-trace-bound}
\end{align}

Now set
\begin{equation}
    A\coloneqq G(\rho)
\end{equation}
and
\begin{equation}
    B\coloneqq G(\sigma).
\end{equation}
Equation~\eqref{eq:DB-vector-field-reflection-form} gives
\begin{equation}
    \begin{split}
        \cV(\rho)-\cV(\sigma)
        ={}&
        \frac{1}{2}
        \left(
            R_\rho(A-B)R_\rho-(A-B)
        \right)
        \\
        &+
        \frac{1}{2}
        \left(
            R_\rho B R_\rho-R_\sigma B R_\sigma
        \right).
    \end{split}
\end{equation}
By Eq.~\eqref{eq:rank-one-reflection-trace-bound} and the
Lipschitz continuity of $G$,
\begin{align}
    \norm{
        \frac{1}{2}
        \left(
            R_\rho(A-B)R_\rho-(A-B)
        \right)
    }_1
    &\le
    2\norm{A-B}_{\rm op}
    \\
    &\le
    2L_G\norm{\rho-\sigma}_1.
    \label{eq:DB-generator-difference-bound}
\end{align}

For the second term, we use
\begin{equation}
    R_\rho B R_\rho-R_\sigma B R_\sigma
    =
    (R_\rho-R_\sigma)BR_\rho
    +
    R_\sigma B(R_\rho-R_\sigma).
\end{equation}
Since $R_\rho$ and $R_\sigma$ are unitary,
$\norm{B}_{\rm op}\le M_G$, and
\begin{equation}
    \norm{R_\rho-R_\sigma}_1
    =
    2\norm{\rho-\sigma}_1,
\end{equation}
we have
\begin{align}
    &\norm{
        \frac{1}{2}
        \left(
            R_\rho B R_\rho-R_\sigma B R_\sigma
        \right)
    }_1
    \notag\\
    \le
    &\frac{1}{2}
    \left(
        \norm{(R_\rho-R_\sigma)BR_\rho}_1
        +
        \norm{R_\sigma B(R_\rho-R_\sigma)}_1
    \right)
    \\
    \le
    &M_G\norm{R_\rho-R_\sigma}_1
    \\
    =
    &2M_G\norm{\rho-\sigma}_1.
    \label{eq:DB-projector-difference-bound}
\end{align}
Combining Eqs.~\eqref{eq:DB-generator-difference-bound} and
\eqref{eq:DB-projector-difference-bound}, we obtain
\begin{equation}
    \norm{\cV(\rho)-\cV(\sigma)}_1
    \le
    \lambda_{\rm exp} \norm{\rho-\sigma}_1.
    \label{eq:DB-vector-field-Lipschitz}
\end{equation}

By Appendix~\ref{appendix:purity-preservation},
Eq.~\eqref{eq:state-dependent-db} preserves all eigenvalues of $\rho(t)$.
In particular, a solution with
$\rho(0)\in\mathcal S_{\rm pure}$ remains in
$\mathcal S_{\rm pure}$ throughout its interval of existence.
Because $\mathcal{S}_{\rm pure}$ is compact and the vector field is
Lipschitz there, the flow exists uniquely for all finite times.

Let
\begin{equation}
    \rho(t)\coloneqq\varphi_t(\rho)
\end{equation}
and
\begin{equation}
    \sigma(t)\coloneqq\varphi_t(\sigma).
\end{equation}
Using the integral form of the two differential equations together
with Eq.~\eqref{eq:DB-vector-field-Lipschitz}, we find
\begin{align}
    \norm{\rho(t)-\sigma(t)}_1
    \le{}&
    \norm{\rho-\sigma}_1
    \notag\\
    &+
    \lambda_{\rm exp}
    \int_0^t
        \norm{\rho(s)-\sigma(s)}_1
    \,ds.
\end{align}
Gronwall's inequality then yields
Eq.~\eqref{eq:generic-exponential-solution-bound}.
\qed

\section{Proof of Lemma~\ref{lem:discretization-error}} 
\label{appendix:discretization-error}
We prove Lemma~\ref{lem:discretization-error}.
\vspace{1em}
\newline
\textit{Proof of Lemma~\ref{lem:discretization-error}}
\hspace{1em}
Fix $\rho\in\mathcal S_{\rm pure}$ and $0<\tau\le1$.
Define
\begin{equation}
    \cV(\sigma)
    \coloneqq
    \comm{\comm{\sigma}{G(\sigma)}}{\sigma}
\end{equation}
for $\sigma\in\mathcal S_{\rm pure}$.
We estimate the deviations of both $\varphi_\tau(\rho)$ and
$\Phi_\tau(\rho)$ from the common first-order expression
$\rho+\tau\cV(\rho)$.

We first record bounds used below.
For a pure state $\sigma \in \mathcal{S}_{\rm pure}$, set
\begin{equation}
    A_\sigma\coloneqq\comm{\sigma}{G(\sigma)}
\end{equation}
and
\begin{equation}
    R_\sigma\coloneqq\id-2\sigma.
\end{equation}
Since $R_\sigma$ is unitary and
$A_\sigma=-\comm{R_\sigma}{G(\sigma)}/2$, we have
\begin{align}
    \norm{A_\sigma}_{\rm op}
    &\le\frac12\left(
        \norm{R_\sigma G(\sigma)}_{\rm op}
        +\norm{G(\sigma)R_\sigma}_{\rm op}
    \right) \\
    &\le\norm{G(\sigma)}_{\rm op}
    \le M_G.
\end{align}
Using $\norm{\sigma}_1=1$ and
$\norm{\comm{A}{B}}_1\le2\norm{A}_{\rm op}\norm{B}_1$,
we obtain
\begin{equation}
    \norm{\cV(\sigma)}_1
    =\norm{\comm{A_\sigma}{\sigma}}_1
    \le2\norm{A_\sigma}_{\rm op}\norm{\sigma}_1
    \le2M_G.
\end{equation}
Moreover, Eq.~\eqref{eq:DB-vector-field-Lipschitz} in
Appendix~\ref{appendix:Lem-generic-exponential-solution-bound}
gives
\begin{equation}
    \norm{\cV(\sigma)-\cV(\omega)}_1
    \le\lambda_{\rm exp}\norm{\sigma-\omega}_1,
\end{equation}
for $\lambda_{\rm exp}=2(M_G+L_G)$ and all $\sigma,\omega\in\mathcal S_{\rm pure}$.

For the exact solution, write $\rho_s=\varphi_s(\rho)$.
By Appendix~\ref{appendix:purity-preservation},
$\rho_s\in\mathcal S_{\rm pure}$ for $0\le s\le\tau$,
so the preceding bounds apply along the entire trajectory.
The integral form of Eq.~\eqref{eq:state-dependent-db} is
\begin{equation}
    \rho_s-\rho=\int_0^s\cV(\rho_u)\,du.
\end{equation}
It follows that
\begin{equation}
    \norm{\rho_s-\rho}_1
    \le\int_0^s\norm{\cV(\rho_u)}_1\,du
    \le2M_Gs.
\end{equation}
Subtracting $\tau\cV(\rho)$ from the integral expression
at $s=\tau$ gives
\begin{equation}
    \varphi_\tau(\rho)-\rho-\tau\cV(\rho)
    =\int_0^\tau
        \left(\cV(\rho_s)-\cV(\rho)\right)\,ds.
\end{equation}
The Lipschitz estimate therefore yields
\begin{align}
    \norm{\varphi_\tau(\rho)-\rho-\tau\cV(\rho)}_1
    &\le\int_0^\tau
        \norm{\cV(\rho_s)-\cV(\rho)}_1\,ds \\
    &\le\lambda_{\rm exp}
        \int_0^\tau\norm{\rho_s-\rho}_1\,ds \\
    &\le2M_G\lambda_{\rm exp}\int_0^\tau s\,ds \\
    &=M_G\lambda_{\rm exp}\tau^2.
\end{align}

We next consider the approximation
\begin{equation}
    \Phi_s(\rho)=e^{sA_\rho}\rho e^{-sA_\rho}.
\end{equation}
Here, $A_\rho$ is held fixed as $s$ varies.
Since $A_\rho$ is anti-Hermitian, $e^{sA_\rho}$ is unitary,
and hence $\norm{\Phi_s(\rho)}_1=1$.
Differentiating with respect to $s$ gives
\begin{align}
    \frac{d}{ds}\Phi_s(\rho)
    &=\comm{A_\rho}{\Phi_s(\rho)},\\
    \frac{d^2}{ds^2}\Phi_s(\rho)
    &=\comm{A_\rho}{\comm{A_\rho}{\Phi_s(\rho)}}.
\end{align}
In particular, $\Phi_0(\rho)=\rho$ and
$\left.d\Phi_s(\rho)/ds\right|_{s=0}=\cV(\rho)$.
Applying the commutator bound twice gives
\begin{align}
    \norm{\frac{d^2}{ds^2}\Phi_s(\rho)}_1
    &\le2\norm{A_\rho}_{\rm op}
        \norm{\comm{A_\rho}{\Phi_s(\rho)}}_1\\
    &\le4\norm{A_\rho}_{\rm op}^2
        \norm{\Phi_s(\rho)}_1\\
    &\le4M_G^2.
\end{align}
Taylor's formula with integral remainder now reads
\begin{equation}
    \Phi_\tau(\rho)-\rho-\tau\cV(\rho)
    =\int_0^\tau(\tau-s)
        \frac{d^2}{ds^2}\Phi_s(\rho)\,ds.
\end{equation}
Consequently,
\begin{align}
    \norm{\Phi_\tau(\rho)-\rho-\tau\cV(\rho)}_1
    &\le4M_G^2\int_0^\tau(\tau-s)\,ds \\
    &=2M_G^2\tau^2.
\end{align}

Finally, the triangle inequality gives
\begin{align}
    \norm{\Phi_\tau(\rho)-\varphi_\tau(\rho)}_1
    &\le\norm{\Phi_\tau(\rho)-\rho-\tau\cV(\rho)}_1\notag\\
    &\quad+\norm{\varphi_\tau(\rho)-\rho-\tau\cV(\rho)}_1 \\
    &\le\left(2M_G^2+M_G\lambda_{\rm exp}\right)\tau^2 \\
    &=C_{\rm disc}\tau^2.
\end{align}
This proves Eq.~\eqref{eq:one-step-discretization-error}.
\qed

\section{Proof of Lemma~\ref{lem:one-step-implementation-error}} 
\label{appendix:proof-one-step-implementation-error}
We prove Lemma~\ref{lem:one-step-implementation-error}.
\vspace{1em}
\newline
\textit{Proof of Lemma~\ref{lem:one-step-implementation-error}}
\hspace{1em}
Fix $\rho\in\mathcal{S}_{\rm pure}$ and define
\begin{equation}
    \begin{split}
        \widehat V(\rho,s)
        \coloneqq{}&
        e^{iG(\rho)s}e^{i\rho s}
        e^{-iG(\rho)s}e^{-i\rho s}
        \\
        &\times
        e^{-iG(\rho)s}e^{-i\rho s}
        e^{iG(\rho)s}e^{i\rho s}.
    \end{split}
\end{equation}
For $s=\sqrt{\tau/2}$, replacing the four exponentials of
$G(\rho)$ by the corresponding $A_G$ factors gives
\begin{equation}
    \norm{
        V_{\rm Gen}(\rho,\tau)-\widehat V(\rho,s)
    }_{\rm op}
    \le
    2\epsilon_G(s)+2\epsilon_G(-s)
    \le 4c_G\tau^2,
\end{equation}
where we repeatedly applied the triangle inequality and used
the unitarity of all factors.

We next compare $\widehat V(\rho,s)$ with $W(\rho,\tau)$.
Differentiating the product at $s=0$ gives
\begin{align}
    \widehat V(\rho,0)
    &=\id,
    \\
    \left.
        \frac{\partial}{\partial s}\widehat V(\rho,s)
    \right|_{s=0}
    &=0,
    \\
    \left.
        \frac{\partial^2}{\partial s^2}\widehat V(\rho,s)
    \right|_{s=0}
    &=4\comm{\rho}{G(\rho)},
    \\
    \left.
        \frac{\partial^3}{\partial s^3}\widehat V(\rho,s)
    \right|_{s=0}
    &=0.
\end{align}

For real $s$, every exponential factor is unitary.
The product rule and the multinomial theorem therefore imply
\begin{align}
    \norm{
        \frac{\partial^4}{\partial s^4}\widehat V(\rho,s)
    }_{\rm op}
    &\le
    \left(
        4\norm{G(\rho)}_{\rm op}
        +4\norm{\rho}_{\rm op}
    \right)^4 \\
    &\le
    4^4(1+M_G)^4.
\end{align}
Taylor's formula with integral remainder, together with
$2s^2=\tau$, consequently yields
\begin{equation}
    \begin{split}
        &\norm{
            \widehat V(\rho,s)
            -\id-\tau\comm{\rho}{G(\rho)}
        }_{\rm op}
        \\
        \le&
        \frac{s^4}{4!}\,4^4(1+M_G)^4
        =
        \frac{8}{3}(1+M_G)^4\tau^2.
    \end{split}
\end{equation}

On the other hand, the reflection $R_\rho=\id-2\rho$ is
unitary, and hence
\begin{equation}
    \norm{\comm{\rho}{G(\rho)}}_{\rm op}
    =
    \frac{1}{2}
    \norm{\comm{R_\rho}{G(\rho)}}_{\rm op}
    \le
    \norm{G(\rho)}_{\rm op}
    \le M_G.
\end{equation}
Since $\comm{\rho}{G(\rho)}$ is anti-Hermitian,
Taylor's formula with integral remainder also gives
\begin{align}
    &\norm{
        W(\rho,\tau)
        -\id-\tau\comm{\rho}{G(\rho)}
    }_{\rm op}
    \notag\\
    \le
    &\int_0^\tau
        (\tau-t)
        \norm{\comm{\rho}{G(\rho)}}_{\rm op}^2
    \,dt
    \\
    \le
    &\frac{1}{2}M_G^2\tau^2.
\end{align}
Combining these estimates by the triangle inequality, we obtain
\begin{equation}
    \begin{split}
        &\norm{
            V_{\rm Gen}(\rho,\tau)-W(\rho,\tau)
        }_{\rm op}
        \\
        \le
        &\left(
            \frac{8}{3}(1+M_G)^4
            +\frac{1}{2}M_G^2
            +4c_G
        \right)\tau^2
        =C_{\rm GC}\tau^2,
    \end{split}
\end{equation}
which proves Eq.~\eqref{eq:symmetric-gc-error-bound}.
\qed

\section{Exponential Contraction for Normalized
Imaginary-Time Evolution}
\label{appendix:ITE-contraction}
We consider the normalized imaginary-time evolution
targeted by DB-QITE~\cite{Gluza2026DBQITE}, which corresponds
to $G(\rho)=H$ for a time-independent Hermitian
Hamiltonian $H$.
The exact solution map is
\begin{equation}
    \varphi_t(\rho)
    =
    \frac{e^{-tH}\rho e^{-tH}}
         {\Tr(e^{-2tH}\rho)}.
    \label{eq:ITE-exact-solution}
\end{equation}
We now derive an exponential contraction bound
on a subset of $\mathcal{S}_{\rm pure}$.
We then apply this bound to the global error and query
analysis of Sec.~\ref{sec:global-analysis} to derive a singly
exponential initial state preparation oracle query upper
bound of $\exp(O(T))$ in the target evolution time $T$,
under the additional conditions stated below.

Suppose that $H$ has a unique ground state with energy
$E_0$ and first excited-state energy $E_1>E_0$.
Let $\Omega\subset\mathcal{S}_{\rm pure}$ be the subset
consisting of all pure states satisfying
$\Tr[(H-E_0\id)\rho]\le(E_1-E_0)/4$, 
i.e., $\Omega \coloneqq \{ \rho \in \mathcal{S}_{\rm pure} \mid \Tr[(H-E_0\id)\rho]\le(E_1-E_0)/4 \}$.
We show that
\begin{equation}
    \norm{\varphi_t(\rho)-\varphi_t(\sigma)}_1
    \le
    e^{-(E_1-E_0)t/2}\norm{\rho-\sigma}_1
    \label{eq:ITE-contraction}
\end{equation}
for all $\rho,\sigma\in\Omega$ and $t\ge0$.
Thus, in Definition~\ref{def:Lipschitz-bound}, we may take
$\Lambda(t)=e^{-\lambda t}$ with
$\lambda=(E_1-E_0)/2>0$.

First, we show that exact solutions initialized in
$\Omega$ remain in $\Omega$.
For $\rho(t)=\varphi_t(\rho)$,
Eq.~\eqref{eq:fixed-db} and $\rho(t)^2=\rho(t)$ give
\begin{equation}
    \frac{d}{dt}\rho(t)
    =
    -H\rho(t)-\rho(t)H
    +2\Tr(H\rho(t))\rho(t).
    \label{eq:ITE-projector-evolution}
\end{equation}
Consequently,
\begin{equation}
    \frac{d}{dt}\Tr(H\rho(t))
    =
    -2\left[
        \Tr(H^2\rho(t))
        -\bigl(\Tr(H\rho(t))\bigr)^2
    \right]
    \le0.
    \label{eq:ITE-energy-decrease}
\end{equation}
The energy condition defining $\Omega$ is therefore
preserved by the exact evolution.

Next, let $\rho(t)=\varphi_t(\rho)$ and
$\sigma(t)=\varphi_t(\sigma)$ for distinct
$\rho,\sigma\in\Omega$.
Since $e^{-tH}$ is invertible at every finite time,
Eq.~\eqref{eq:ITE-exact-solution} implies that
$\rho(t)\ne\sigma(t)$.
Let $P(t)$ denote the rank-two projector onto the
subspace spanned by their state vectors.
For two distinct pure states,
\begin{equation}
    (\rho(t)-\sigma(t))^2
    =
    \bigl[1-\Tr(\rho(t)\sigma(t))\bigr]P(t),
    \label{eq:ITE-projector-difference}
\end{equation}
and
\begin{equation}
    \norm{\rho(t)-\sigma(t)}_1^2
    =
    4\bigl[1-\Tr(\rho(t)\sigma(t))\bigr].
    \label{eq:ITE-pure-state-distance}
\end{equation}
Differentiating
Eq.~\eqref{eq:ITE-pure-state-distance}
and using
Eqs.~\eqref{eq:ITE-projector-evolution}
and~\eqref{eq:ITE-projector-difference} yields
\begin{equation}
    \begin{split}
        \frac{d}{dt}\log\norm{\rho(t)-\sigma(t)}_1
        &=
        \Tr(H\rho(t))+\Tr(H\sigma(t))\\
        &\quad-\Tr(HP(t)).
    \end{split}
    \label{eq:ITE-distance-rate}
\end{equation}
Since $P(t)$ has rank two,
$\Tr(HP(t))\ge E_0+E_1$.
Moreover, $\rho(t),\sigma(t)\in\Omega$ implies
\begin{equation}
    \Tr(H\rho(t))+\Tr(H\sigma(t))
    \le
    2E_0+\frac{E_1-E_0}{2}.
\end{equation}
Substituting these inequalities into
Eq.~\eqref{eq:ITE-distance-rate} gives
\begin{equation}
    \frac{d}{dt}\log\norm{\rho(t)-\sigma(t)}_1
    \le
    -\frac{E_1-E_0}{2}.
\end{equation}
Integrating proves Eq.~\eqref{eq:ITE-contraction}.
The same inequality holds trivially when $\rho=\sigma$.

Finally, we apply this exponential contraction bound
to the global error and query analysis of
Sec.~\ref{sec:global-analysis}.
Suppose that $\rho_0\in\Omega$.
Equation~\eqref{eq:ITE-energy-decrease} ensures that
the exact trajectory remains in $\Omega$.
We additionally assume that the states $\rho_k$
generated by the algorithm remain in $\Omega$,
so that Eq.~\eqref{eq:Omega-containing-trajectories}
holds.
Under the remaining hypotheses of
Corollary~\ref{cor:query-complexity-simplified},
at fixed accuracy and model-dependent constants
and with $b=O(1)$ uniformly at the required step
sizes and accuracies, we obtain
\begin{equation}
    Q_M\le\exp(O(T)).
\end{equation}
Thus, under these additional conditions, the
exponential contraction bound improves the general
doubly exponential initial state preparation oracle
query upper bound $\exp(\exp(O(T)))$ to a singly
exponential upper bound in the target evolution
time $T$.

\section{Discrete GP Equation} 
\label{appendix:discrete-GP}
Equation~\eqref{eq:discrete-GP} is an instance of
Eq.~\eqref{eq:NLSE} with $H(\rho)=K+g\cD(\rho)$.
By Eq.~\eqref{eq:generator-NLSE}, the density operator
$\rho(t)$ satisfies the double-bracket equation
\eqref{eq:state-dependent-db} with
\begin{equation}
    G(\rho)
    \coloneqq
    -i\comm{\rho}{K+g\cD(\rho)}.
    \label{eq:generator-GP}
\end{equation}
Lemma~\ref{lem:GP-generator-bounds} verifies that $G(\rho)$ satisfies
the boundedness and Lipschitz continuity assumptions
imposed in Sec.~\ref{sec:proposed-framework}.

\begin{lemma}
\label{lem:GP-generator-bounds}
For the generator $G(\rho)$ defined in Eq.~\eqref{eq:generator-GP},
the boundedness and Lipschitz conditions in
Eqs.~\eqref{eq:G-sup-on-pure} and~\eqref{eq:G-Lipschitz-bound}
hold with
\begin{align}
    M_G
    &= \norm{K}_{\rm op}+\frac{\lvert g\rvert}{2},
    \\
    L_G
    &= \norm{K}_{\rm op}+\lvert g\rvert.
\end{align}
\end{lemma}

\begin{proof}
For any $\rho\in\mathcal{S}_{\rm pure}$ and any linear operator $X$ on $\cH$,
the unitarity of $R_\rho=\id-2\rho$ implies
\begin{equation}
    \norm{\comm{\rho}{X}}_{\rm op}
    =
    \frac{1}{2}\norm{\comm{R_\rho}{X}}_{\rm op}
    \le
    \norm{X}_{\rm op}.
\end{equation}
Moreover, $0\le\cD(\rho)\le\id$ gives
\begin{equation}
    \norm{\cD(\rho)-\frac{1}{2}\id}_{\rm op}
    \le
    \frac{1}{2}.
\end{equation}
Since scalar multiples of the identity do not contribute
to commutators, Eq.~\eqref{eq:generator-GP} yields
\begin{align}
    \norm{G(\rho)}_{\rm op}
    &\le
    \norm{\comm{\rho}{K}}_{\rm op}
    +
    \lvert g \rvert
    \norm{
        \comm{\rho}{\cD(\rho)-\frac{1}{2}\id}
    }_{\rm op}
    \\
    &\le
    \norm{K}_{\rm op}+\frac{\lvert g \rvert}{2}
    =M_G.
\end{align}
Thus, Eq.~\eqref{eq:G-sup-on-pure} holds with the stated
choice of $M_G$.

Next, let $\rho,\sigma\in\mathcal{S}_{\rm pure}$.
The operator $\rho-\sigma$ is Hermitian, traceless, and
of rank at most two. Its nonzero eigenvalues therefore
have equal magnitude and opposite signs, so
\begin{equation}
    \norm{\rho-\sigma}_{\rm op}
    =
    \frac{1}{2}\norm{\rho-\sigma}_1.
\end{equation}
The definition of $\cD$ also gives
\begin{equation}
    \norm{\cD(\rho-\sigma)}_{\rm op}
    =
    \max_j
    \lvert
        \bra{j}(\rho-\sigma)\ket{j}
    \rvert
    \le
    \norm{\rho-\sigma}_{\rm op}.
\end{equation}
Using the linearity of $\cD$, we write
\begin{equation}
    \begin{split}
        G(\rho)-G(\sigma)
        ={}&
        -i\comm{
            \rho-\sigma
        }{
            K+g\left(\cD(\rho)-\frac{1}{2}\id\right)
        }
        \\
        &-ig\comm{\sigma}{\cD(\rho-\sigma)}.
    \end{split}
\end{equation}
Applying the preceding bounds to the two commutators yields
\begin{align}
    &\norm{G(\rho)-G(\sigma)}_{\rm op}
    \notag\\
    \le
    &2\norm{\rho-\sigma}_{\rm op}
    \left(
        \norm{K}_{\rm op}+\frac{\lvert g \rvert}{2}
    \right) \notag
    \\
    &+
    \lvert g \rvert\norm{\cD(\rho-\sigma)}_{\rm op}
    \\
    \le
    &2\left(\norm{K}_{\rm op}+\lvert g \rvert\right)
    \norm{\rho-\sigma}_{\rm op}
    \\
    =
    &L_G\norm{\rho-\sigma}_1.
\end{align}
Hence, Eq.~\eqref{eq:G-Lipschitz-bound} holds with the stated
choice of $L_G$, completing the proof.
\end{proof}

As in Sec.~\ref{sec:proposed-framework}, we assume access
to an initial state preparation oracle $U_0$ satisfying
$U_0\ket{0}=\ket{\psi(0)}$ and its inverse $U_0^\dagger$.
We additionally assume ancilla-free access to $e^{iK\theta}$
for arbitrary $\theta\in\mathbb{R}$.
At step $k$, the recursively constructed circuit $U_k$
prepares $\ket{\psi_k}$, with $\rho_k=\proj{\psi_k}$.
For $0<\tau\le1$, the NLSE-specific circuit of Sec.~\ref{subsec:one-step-error-NLSE} approximates
$W(\rho_k,\tau)=e^{\tau\comm{\rho_k}{G(\rho_k)}}$
using $O(\tau^2)$-accurate approximations to
$e^{\pm i(K+g\cD(\rho_k))\tau/2}$.
Since $e^{\pm iK\tau/2}$ is available, applying the
first-order Lie--Trotter formula reduces this task to
approximating $e^{\pm ig\cD(\rho_k)\tau/2}$ with
$O(\tau^2)$ operator-norm error.

To implement these exponentials, we use the following
representation of the dephasing map as a uniform average
over conjugations by tensor products of $\id$ and $Z$:
\begin{equation}
    \cD(\rho)
    =
    \frac{1}{2^n}
    \sum_{s\in\mathbb{F}_2^n}Z^s\rho Z^s.
\end{equation}
Here, $Z$ denotes the Pauli $Z$ operator,
$\mathbb{F}_2=\{0,1\}$ is the field with two elements,
and $Z^s\coloneqq Z^{s_1}\otimes\cdots\otimes Z^{s_n}$.
Applying the first-order Lie--Trotter formula directly to this
sum gives an $O(\tau^2)$-accurate approximation with $N=2^n$ factors.
Each factor can be implemented as
$e^{\pm ig Z^s\rho_k Z^s\tau/(2N)}
= Z^s U_k e^{\pm ig\proj{0}\tau/(2N)}U_k^\dagger Z^s$,
so the approximation requires $O(N)$ calls to
$U_k$ and $U_k^\dagger$.
Theorem~\ref{thm:query-complexity-Lipschitz-bounds} with $b=O(N)$
then yields an initial state preparation oracle query upper
bound of $\exp(O(Mn))$, which is exponential in $n$ even for fixed $M$ and does not establish a polylogarithmic
dependence on $N$.
To avoid this exponential dependence on $n$, we seek an
alternative implementation 
for which we may take
$b=O(\poly(n))$ in Theorem~\ref{thm:query-complexity-Lipschitz-bounds}
for fixed $T$ and $\epsilon$,
assuming that $\lvert g\rvert$ and $\norm{K}_{\rm op}$
are bounded independently of $n$.
To this end, we approximate the dephasing map using a small-bias set.

\begin{definition}[$\delta$-biased set;
Definition 2.2 in Ref.~\cite{ta-shma_explicit_2017}]
\label{def:delta-biased-set}
Let $\mathcal{S}=(s_1,\ldots,s_L)$ be an indexed multiset
of $L$ vectors in $\mathbb{F}_2^n$, each consisting of $n$ bits.
For $z\in\mathbb{F}_2^n\setminus\{0\}$, define
\begin{equation}
    \beta_{\mathcal{S}}(z)
    \coloneqq
    \frac{1}{L}
    \sum_{\nu=1}^{L}(-1)^{s_\nu\cdot z},
\end{equation}
where the inner product is taken over $\mathbb{F}_2$.
We call $\mathcal{S}$ $\delta$-biased if
$\lvert\beta_{\mathcal{S}}(z)\rvert\le\delta$
for every $z\in\mathbb{F}_2^n\setminus\{0\}$.
\end{definition}

The following theorem provides an explicit construction
with a controlled number of terms.

\begin{theorem}[Theorem 1.2 in Ref.~\cite{ta-shma_explicit_2017}]
\label{thm:near-optimal-small-bias}
For every integer $n\ge1$ and every $0<\delta<1/2$,
there exists an explicit, deterministically constructible
$\delta$-biased indexed multiset $\mathcal{S}$ in
$\mathbb{F}_2^n$ of size
\begin{equation}
    L
    =
    O\left(\frac{n}{\delta^{2+o(1)}}\right).
\end{equation}
Here, $o(1)$ denotes a quantity tending to zero as
$\delta\to0$.
\end{theorem}

For a $\delta$-biased indexed multiset $\mathcal{S}$ of size $L$
provided by Theorem~\ref{thm:near-optimal-small-bias},
we define the approximate dephasing map
\begin{equation}
    \cD_{\mathcal{S}}(\rho)
    \coloneqq
    \frac{1}{L}
    \sum_{\nu=1}^{L}Z^{s_\nu}\rho Z^{s_\nu}.
\end{equation}
We use this approximate dephasing map to construct the
required unitary exponentials.
Lemma~\ref{lem:small-bias-dephasing-error} bounds the operator-norm error
of this approximation for pure states.

\begin{lemma}[Small-bias approximation of dephasing]
\label{lem:small-bias-dephasing-error}
For $\rho\in\mathcal{S}_{\rm pure}$,
\begin{equation}
    \norm{\cD_{\mathcal{S}}(\rho)-\cD(\rho)}_{\rm op}
    \le\delta.
\end{equation}
\end{lemma}

\begin{proof}
For distinct $x,y\in\mathbb{F}_2^n$,
\begin{equation}
    \braket{x|\cD_{\mathcal{S}}(\rho)-\cD(\rho)|y}
    =
    \beta_{\mathcal{S}}(x\oplus y)\braket{x|\rho|y},
\end{equation}
while the diagonal entries vanish.
Since $\lvert\beta_{\mathcal{S}}(x\oplus y)\rvert\le\delta$,
\begin{align}
    \norm{\cD_{\mathcal{S}}(\rho)-\cD(\rho)}_{\rm op}^2
    \le&
    \norm{\cD_{\mathcal{S}}(\rho)-\cD(\rho)}_{\rm F}^2
    \\
    \le&
    \delta^2\sum_{x\ne y}
    \left\lvert\braket{x|\rho|y}\right\rvert^2
    \\
    \le&
    \delta^2\Tr(\rho^2)
    =
    \delta^2,
\end{align}
where $\norm{\cdot}_{\rm F}$ denotes the Frobenius norm.
\end{proof}

We now construct an ancilla-free approximation to
$e^{\pm ig\cD(\rho_k)\tau/2}$ with $O(\tau^2)$ operator-norm error.
For fixed $\tau$, this implementation uses $O(\poly(n))$
calls in total to $U_k$ and $U_k^\dagger$.
First, we choose
\begin{equation}
    \delta=\frac{\tau}{4(1+\lvert g\rvert)}.
\end{equation}
This choice satisfies $0<\delta<1/2$ for $0<\tau\le1$.
Since $\cD_{\mathcal{S}}(\rho_k)$ and $\cD(\rho_k)$ are Hermitian,
Lemma~\ref{lem:small-bias-dephasing-error} implies,
for either choice of sign,
\begin{equation}
    \norm{
        e^{\pm ig\cD_{\mathcal{S}}(\rho_k)\tau/2}
        -
        e^{\pm ig\cD(\rho_k)\tau/2}
    }_{\rm op}
    \le
    \frac{\lvert g\rvert}{2}\tau\delta
    \le
    \frac{\tau^2}{8}.
    \label{eq:exp-small-bias-approx}
\end{equation}
For this choice of $\delta$,
Theorem~\ref{thm:near-optimal-small-bias} provides
a $\delta$-biased indexed multiset $\mathcal{S}$ of size
\begin{equation}
    L
    =
    O\left(
        n\left(\frac{1+\lvert g\rvert}{\tau}\right)^{2+o(1)}
    \right).
    \label{eq:near-optimal-set-for-GP}
\end{equation}
Note that the multiset $\mathcal{S}$ is chosen for each step size $\tau$.
We next approximate $e^{ig\cD_{\mathcal{S}}(\rho_k)\theta}$
for $\theta \in \mathbb{R}$ using the first-order Lie--Trotter formula:
\begin{equation}
    A_{\mathcal{S},k}(\theta)
    \coloneqq
    \prod_{\nu=1}^{L}
    e^{i(g/L)\rho_{k,s_\nu}\theta},
\end{equation}
where $\rho_{k,s_\nu}\coloneqq Z^{s_\nu}\rho_k Z^{s_\nu}$.
Each factor is implemented without ancilla qubits as
\begin{equation}
    e^{i(g/L)\rho_{k,s_\nu}\theta}
    =
    Z^{s_\nu}U_k
    e^{i(g/L)\proj{0}\theta}
    U_k^\dagger Z^{s_\nu}.
\end{equation}
Thus, $A_{\mathcal{S},k}(\theta)$ uses $2L$ calls in total
to $U_k$ and $U_k^\dagger$.
The first-order Lie--Trotter formula bound gives
\begin{align}
    \norm{
        A_{\mathcal{S},k}(\theta)
        -
        e^{ig\cD_{\mathcal{S}}(\rho_k)\theta}
    }_{\rm op}
    &\le
    \frac{g^2\theta^2}{2L^2}
    \sum_{1\le\mu<\nu\le L}
    \norm{
        \comm{\rho_{k,s_\mu}}{\rho_{k,s_\nu}}
    }_{\rm op}
    \\
    &\le
    \frac{g^2 \theta^2}{2},
\end{align}
where we used $\norm{\rho_{k,s_\nu}}_{\rm op}=1$.
Setting $\theta=\pm\tau/2$ and combining this bound with
Eq.~\eqref{eq:exp-small-bias-approx} yields
\begin{equation}
    \norm{
        A_{\mathcal{S},k}(\pm\tau/2)
        -
        e^{\pm ig\cD(\rho_k)\tau/2}
    }_{\rm op}
    \le
    \frac{g^2+1}{8}\tau^2.
\end{equation}
We therefore implement the $A_H$ factors in
Eq.~\eqref{eq:NLSE-one-step} as
\begin{equation}
    A_H(\rho_k, \pm \tau/2)
    =
    e^{\pm iK \tau/2}A_{\mathcal{S},k}(\pm \tau/2).
\end{equation}
Applying the first-order Lie--Trotter formula and using
$\norm{\cD(\rho_k)}_{\rm op}\le1$, we obtain
\begin{equation}
    \begin{split}
        &\norm{
            A_H(\rho_k,\pm\tau/2)
            -
            e^{\pm i(K+g\cD(\rho_k))\tau/2}
        }_{\rm op}
        \\
        \le&
        \left(
            \frac{g^2+1}{8}
            +
            \frac{\lvert g\rvert\norm{K}_{\rm op}}{4}
        \right)\tau^2.
    \end{split}
\end{equation}
These unitaries meet the accuracy requirement of
Sec.~\ref{subsec:one-step-error-NLSE}, with $2L$ calls in total to
$U_k$ and $U_k^\dagger$ for each sign.

Finally, we count queries to $U_0$ and $U_0^\dagger$
by recursively expanding the calls to $U_k$ and $U_k^\dagger$.
The circuit $V_{\rm NLSE}(\rho_k,\tau)$ in
Eq.~\eqref{eq:NLSE-one-step} contains two $A_H$ factors
and two reflections $R_{\rho_k}$.
The $A_H$ factors require $4L$ calls in total, and the
reflections require four additional calls.
Including the rightmost $U_k$ in
$U_{k+1}=V_{\rm NLSE}(\rho_k,\tau)U_k$ gives
\begin{equation}
    Q_{k+1}\le(4L+5)Q_k.
\end{equation}
Using $Q_0=1$, $M=T/\tau$, and
Eq.~\eqref{eq:near-optimal-set-for-GP}, we obtain
\begin{align}
    Q_M
    &\le
    (4L+5)^M
    \\
    &\le
    \exp\left[
        O\left(
            \frac{T}{\tau}
            \log\frac{2n(1+\lvert g\rvert)}{\tau}
        \right)
    \right].
\end{align}
Thus, the simulation is deterministic, fully coherent,
and ancilla-free.
Here $b=4L+4$ in Theorem~\ref{thm:query-complexity-Lipschitz-bounds}.
For fixed $T$ and $\epsilon$, with $\lvert g\rvert$ and
$\norm{K}_{\rm op}$ bounded independently of $n$,
the one-step error bounds and
Lemma~\ref{lem:generic-exponential-solution-bound}
allow $M=T/\tau$ to be chosen independently of $n$
to achieve the target accuracy.
Equation~\eqref{eq:near-optimal-set-for-GP} then gives $L=O(n)$,
and hence $Q_M=O(n^M)=O(\poly(n))$.
Under these assumptions, the small-bias construction thus
reduces the initial state preparation oracle query upper bound
from exponential to polynomial in $n$, yielding an upper bound that is polylogarithmic in $N$.

\section{Proof of Lemma~\ref{lem:single-qubit-GP-relative-difference}} 
\label{appendix:proof-single-qubit-GP-relative-difference}
We prove Lemma~\ref{lem:single-qubit-GP-relative-difference}.
\vspace{1em}
\newline
\textit{Proof of Lemma~\ref{lem:single-qubit-GP-relative-difference}}
\hspace{1em}
We first prove 
the exponential expansion bound.
With $s=gt$ and $r=(x,y,z)$, the Bloch equations become
\begin{equation}
    \frac{dr}{ds}=f(r),
\end{equation}
where
\begin{equation}
    f(x,y,z)=\left(-yz,\,z\left(x-\frac12\right),\,\frac y2\right).
\end{equation}
Since $r\cdot f(r)=0$, solutions preserve $\norm{r}_2$.
The Jacobian is
\begin{equation}
    Df(r)=
    \begin{pmatrix}
        0&-z&-y\\
        z&0&x-1/2\\
        0&1/2&0
    \end{pmatrix}.
\end{equation}
For $\norm{r}_2\le1$ and $v=(v_x,v_y,v_z)\in\mathbb R^3$,
the Cauchy--Schwarz inequality gives
\begin{align}
    v^{\mathsf T}Df(r)v
    &=v_z(-yv_x+xv_y) \\
    &\le\sqrt{x^2+y^2}\,|v_z|\sqrt{v_x^2+v_y^2} \\
    &\le\frac12\norm{v}_2^2.
\end{align}
Let $r(s)$ and $q(s)$ be two solution Bloch vectors and set
$d(s)=r(s)-q(s)$.
The line segment joining them lies in the closed unit ball,
so the mean value formula yields
\begin{equation}
    f(r(s))-f(q(s))
    =\int_0^1Df(q(s)+u d(s))d(s)\,du.
\end{equation}
Consequently,
\begin{equation}
    \frac{d}{ds}\norm{d(s)}_2^2
    =2\int_0^1d(s)^{\mathsf T}
       Df(q(s)+u d(s))d(s)\,du
    \le\norm{d(s)}_2^2.
\end{equation}
Gronwall's inequality gives
$\norm{d(s)}_2\le e^{s/2}\norm{d(0)}_2$.
For qubit states $\rho$ and $\sigma$ with Bloch vectors $r$
and $q$, the eigenvalues of $\rho-\sigma$ are
$\pm\norm{r-q}_2/2$.
Hence
\begin{equation}
    \norm{\rho-\sigma}_1=\norm{r-q}_2.
\end{equation}
Setting $s=g\tau$ proves
Eq.~\eqref{eq:single-qubit-GP-solution-bound} for every $\tau\ge0$.

We next prove the one-step error bound relative to
$\rho_{\rm sta}$.
Set $h=g\tau$ and define
\begin{equation}
    \overline H(\rho)
    =\frac X4+\frac{\rho+Z\rho Z}{2}.
\end{equation}
Let $\Psi_h$ be the solution map of
$d\rho/dh=-i\comm{\overline H(\rho)}{\rho}$, so that
$\Psi_h(\rho)=\varphi_{h/g}(\rho)$.
Using Eq.~\eqref{eq:single-qubit-GP-H-step}, write
\begin{equation}
    A_\pm(\rho,h)
    =e^{\pm ihX/8}e^{\pm ih\rho/4}e^{\pm ihZ\rho Z/4}
\end{equation}
and
\begin{equation}
    V_h(\rho)=A_-(\rho,h)R_\rho A_+(\rho,h)R_\rho,
\end{equation}
where $R_\rho=\id-2\rho$.
Then
$\widetilde\Phi^{({\rm GP})}_\tau(\rho)
=V_h(\rho)\rho V_h(\rho)^\dagger$.

At $h=0$, $A_\pm(\rho,0)=\id$ and
\begin{equation}
    \left.\partial_h A_\pm(\rho,h)\right|_{h=0}
    =\pm\frac i2\overline H(\rho).
\end{equation}
For pure $\rho$, the identities $R_\rho^2=\id$ and
$R_\rho\rho=\rho R_\rho=-\rho$ therefore give
$V_0(\rho)=\id$ and
\begin{equation}
    \left.\partial_h V_h(\rho)\right|_{h=0}
    =-\frac i2\left(
        \overline H(\rho)-R_\rho\overline H(\rho)R_\rho
      \right).
\end{equation}
Since
$\comm{R_\rho\overline H(\rho)R_\rho}{\rho}
=-\comm{\overline H(\rho)}{\rho}$,
we obtain
\begin{align}
    &\left.\partial_h
        \left[V_h(\rho)\rho V_h(\rho)^\dagger\right]
      \right|_{h=0}\notag\\
    &\qquad=-i\comm{\overline H(\rho)}{\rho}
      =\left.\partial_h\Psi_h(\rho)\right|_{h=0}.
\end{align}
Thus, the implemented and exact maps agree to first order
in $h$.

Both maps also fix $\rho_{\rm sta}=\proj{+}$ exactly.
Indeed, $Z\rho_{\rm sta}Z=\id-\rho_{\rm sta}$ implies that
$\overline H(\rho_{\rm sta})=X/4+\id/2$ commutes with
$\rho_{\rm sta}$, and hence
$\Psi_h(\rho_{\rm sta})=\rho_{\rm sta}$.
Moreover,
\begin{equation}
    A_\pm(\rho_{\rm sta},h)
    =e^{\pm ih/4}e^{\pm ihX/8}.
\end{equation}
These operators commute with $R_{\rho_{\rm sta}}=-X$,
so $V_h(\rho_{\rm sta})=\id$.

To obtain a bound proportional to the distance from
$\rho_{\rm sta}$, represent any pure qubit state as
\begin{equation}
    \rho(\theta,\phi)
    =\frac{\id+\cos\theta\,X
       +\sin\theta\cos\phi\,Y
       +\sin\theta\sin\phi\,Z}{2},
\end{equation}
where $0\le\theta\le\pi$ and $0\le\phi\le2\pi$.
In particular, $\rho(0,\phi)=\rho_{\rm sta}$.
For $\rho=\rho(\theta,\phi)$, define
\begin{equation}
    E(\theta,\phi,h)
    =V_h(\rho)\rho V_h(\rho)^\dagger-\Psi_h(\rho).
\end{equation}
The first-order agreement gives $E(\theta,\phi,0)=0$ and
$\left.\partial_hE(\theta,\phi,h)\right|_{h=0}=0$.
Taylor's formula with integral remainder yields
\begin{equation}
    E(\theta,\phi,h)=h^2\mathcal R(\theta,\phi,h),
    \label{eq:single-qubit-GP-remainder}
\end{equation}
where
\begin{equation}
    \mathcal R(\theta,\phi,h)
    =\int_0^1(1-u)
      \left.\partial_s^2E(\theta,\phi,s)\right|_{s=uh}\,du.
\end{equation}
This formula also defines $\mathcal R$ at $h=0$.

Fix a finite constant $\Delta_0>0$.
The Bloch-vector field is polynomial and preserves
$\norm{r}_2$, so its solution map depends smoothly on the
initial state and time throughout any finite time interval.
The implemented map is also smooth, being a finite product
of matrix exponentials and reflections.
It follows that $\partial_\theta\mathcal R$ is continuous
on the closed, bounded parameter range, and therefore
\begin{equation}
    L\coloneqq
    \max_{\substack{
        0\le\theta\le\pi,\ 0\le\phi\le2\pi\\
        0\le h\le\Delta_0
    }}
    \norm{\partial_\theta\mathcal R(\theta,\phi,h)}_1
    <\infty.
\end{equation}
Since both maps fix $\rho_{\rm sta}$,
$\mathcal R(0,\phi,h)=0$, including at $h=0$ by continuity.
The fundamental theorem of calculus then gives
\begin{equation}
    \norm{\mathcal R(\theta,\phi,h)}_1
    \le\int_0^\theta
       \norm{\partial_u\mathcal R(u,\phi,h)}_1\,du
    \le L\theta.
\end{equation}
The Bloch-vector distance satisfies
\begin{equation}
    \norm{\rho(\theta,\phi)-\rho_{\rm sta}}_1
    =2\sin\frac\theta2
    \ge\frac{2\theta}{\pi},
\end{equation}
where the last inequality uses
$\sin u\ge2u/\pi$ for $0\le u\le\pi/2$.
Combining these estimates with
Eq.~\eqref{eq:single-qubit-GP-remainder} yields
\begin{equation}
    \norm{V_h(\rho)\rho V_h(\rho)^\dagger-\Psi_h(\rho)}_1
    \le\frac{\pi L}{2}h^2\norm{\rho-\rho_{\rm sta}}_1.
\end{equation}
Setting $h=g\tau$ proves
Eq.~\eqref{eq:single-qubit-GP-relative-difference}, for example
with $C=1+\pi L/2$.
For fixed $\Delta_0$, the maps defining $L$ depend only on
the constant $\Delta_0$ and the state parameters.
Thus, $C$ is independent of $g,T$, and $\xi$.
\qed

\section{Choice of the Number of Time Steps in Sec.~\ref{subsec:single-qubit-GP-error-upper}}
\label{appendix:single-qubit-GP-step-count}
Fix $0<\epsilon\le1$, and let $C,\Delta_0>0$ be the constants
in Lemma~\ref{lem:single-qubit-GP-relative-difference}.
Let $S$ be defined by
Eq.~\eqref{eq:single-qubit-GP-simulation-duration}.
With $\tau=S/M$, Eq.~\eqref{eq:single-qubit-GP-global-error}
gives
\begin{equation}
    \max_{0\le k\le M}e_k^\pm
    \le\frac{32C}{\epsilon}
       \frac{(gS)^2}{M}
       \exp\left(C\frac{(gS)^2}{M}\right),
\end{equation}
provided that $g \tau = gS/M\le\Delta_0$.

To satisfy this condition and make the numerical error
at most $\epsilon/2$, choose
\begin{equation}
    M=
    \left\lceil
    \max\left\{
        1,\,
        \frac{gS}{\Delta_0},\,
        \frac{64C\exp(1)(gS)^2}{\epsilon^2}
    \right\}
    \right\rceil.
\end{equation}
This choice ensures $M\ge1$ and $gS/M\le\Delta_0$.
Moreover,
\begin{equation}
    C\frac{(gS)^2}{M}
    \le\frac{\epsilon^2}{64\exp(1)}
    \le1.
\end{equation}
The exponential factor in the error bound is therefore
at most $\exp(1)$, yielding
\begin{equation}
    \max_{0\le k\le M}e_k^\pm
    \le\frac{32C\exp(1)}{\epsilon}\frac{(gS)^2}{M}
    \le\frac{\epsilon}{2}.
\end{equation}
Equation~\eqref{eq:single-qubit-GP-duration-error} then gives
the accuracy guarantee for all $t \in [S,T]$.

Finally, rounding up increases $M$ by at most one, so
\begin{align}
    M
    &\le2+\frac{gS}{\Delta_0}
      +\frac{64C\exp(1)(gS)^2}{\epsilon^2}\\
    &=O\left(1+gS+\frac{(gS)^2}{\epsilon^2}\right).
\end{align}
Since $S\le T$, this also gives the uniform estimate
$M=O(1+gT+(gT)^2/\epsilon^2)$.
The implied constants depend only on $C$ and $\Delta_0$,
which are independent of $g,T,\xi$, and $\epsilon$.
The same choice of $M$ works for both candidate initial states
and every unitary oracle satisfying either preparation
condition for the fixed known parameters.

\section{Final Angle Adjustment and Query Bound in Sec.~\ref{subsec:single-qubit-GP-achieving-optimal}}
\label{appendix:single-qubit-GP-final-adjustment}
We construct the final angle adjustment used in
Sec.~\ref{subsec:single-qubit-GP-achieving-optimal}
and derive the initial state preparation oracle query bound
in Eq.~\eqref{eq:single-qubit-GP-tighter-upper-bound}.
Recall that $\rho_{\rm sta}=\proj{+}$ and
$\widehat\rho_\pm(u)=\proj{\chi_\pm(u)}$, where
$\ket{\chi_\pm(u)}=\cos u\ket{+}\pm\sin u\ket{-}$.

Let $u$ and $v$ be known angles satisfying
$0<u\le\pi/4$, $0<v\le\pi/4$, and $v\le3u$.
Choose
\begin{equation}
    \theta
    =2\arcsin\sqrt{
       \frac{1+\sin v/\sin u}{4\cos^2u}}.
    \label{eq:single-qubit-GP-angle-change-phase}
\end{equation}
The hypotheses imply
$v\le3u\le3\pi/4\le\pi-v$, and hence
$\sin v\le\sin3u=\sin u(4\cos^2u-1)$.
Therefore,
\begin{equation}
    0<
    \frac{1+\sin v/\sin u}{4\cos^2u}
    \le1,
\end{equation}
so $\theta$ is real and well defined.

To compute the action of the two projector exponentials,
we use
\begin{equation}
    e^{i\theta\widehat\rho_\pm(u)}
    =\id+(e^{i\theta}-1)\widehat\rho_\pm(u)
\end{equation}
and
\begin{equation}
    e^{i\theta\rho_{\rm sta}}\ket{\chi_\pm(u)}
    =\cos u\,e^{i\theta}\ket{+}\pm\sin u\ket{-}.
\end{equation}
Writing $z=\cos^2u\,e^{i\theta}+\sin^2u$, we have
\begin{equation}
    \bra{\chi_\pm(u)}
    e^{i\theta\rho_{\rm sta}}
    \ket{\chi_\pm(u)}
    =z.
\end{equation}
Consequently,
\begin{equation}
    e^{i\theta\widehat\rho_\pm(u)}
    e^{i\theta\rho_{\rm sta}}\ket{\chi_\pm(u)}
    =A\ket{+}\pm B\ket{-},
    \label{eq:single-qubit-GP-angle-change-state}
\end{equation}
where
\begin{align}
    A&=\cos u\left[e^{i\theta}+(e^{i\theta}-1)z\right],\\
    B&=\sin u\left[1+(e^{i\theta}-1)z\right]\\
     &=\sin u\,e^{i\theta}\left[1+2\cos^2u(\cos\theta-1)\right]\\
     &=\sin u\,e^{i\theta}
       \left[1-4\cos^2u\sin^2(\theta/2)\right]\\
     &=-e^{i\theta}\sin v.
\end{align}
The last equality follows from
Eq.~\eqref{eq:single-qubit-GP-angle-change-phase}.

Since the projector exponentials are unitary,
$|A|^2+|B|^2=1$.
Thus, $|A|=\cos v>0$.
Define
\begin{equation}
    D_{u,v}
    =\frac{A^*}{\cos v}\rho_{\rm sta}
      -e^{-i\theta}(\id-\rho_{\rm sta}).
\end{equation}
Both coefficients in this decomposition have modulus one,
so $D_{u,v}$ is unitary.
Moreover,
\begin{align}
    D_{u,v}\left(A\ket{+}\pm B\ket{-}\right)
    &=\cos v\ket{+}\pm\sin v\ket{-}\\
    &=\ket{\chi_\pm(v)}.
\end{align}
Hence $D_{u,v}$ removes both coefficient phases.

If $\widehat U\ket{0}=\ket{\chi_\pm(u)}$ up to a global
phase, then
\begin{equation}
    \widehat Ue^{i\theta\proj{0}}\widehat U^\dagger
    =e^{i\theta\widehat\rho_\pm(u)}.
\end{equation}
Combining this identity with
Eq.~\eqref{eq:single-qubit-GP-angle-change-state} and the
phase correction above proves that
$\mathcal P_{u,v}[\widehat U]$ in
Eq.~\eqref{eq:single-qubit-GP-angle-change-circuit}
prepares $\ket{\chi_\pm(v)}$ up to a global phase.
Its implementation uses three calls in total to
$\widehat U$ and $\widehat U^\dagger$.
The phase $\theta$ and the unitary $D_{u,v}$ depend only
on $u$ and $v$, so the same non-query gates work for both
candidates.
The gates $e^{i\theta\rho_{\rm sta}}$ and $D_{u,v}$ are
diagonal in the basis $\{\ket{+},\ket{-}\}$ and therefore,
up to global phases, are rotations generated by $X$.

We next derive the query bound for the final circuit.
The definition of $m$ in
Eq.~\eqref{eq:single-qubit-GP-tripling-count} ensures
$0<\alpha_m\le\pi/4$ and
$0<\alpha(T)\le\min\{3\alpha_m,\pi/4\}$.
Thus, the adjustment applies with $u=\alpha_m$ and
$v=\alpha(T)$, including when $\alpha(T)<\alpha_0$.

The recursively constructed circuit $\widehat U_m$ uses
$3^m$ queries to $U_0$ and $U_0^\dagger$.
The final adjustment uses three calls to $\widehat U_m$
and $\widehat U_m^\dagger$, while the final rotation
$e^{i\beta(T)X/2}$ requires no such queries.
The total query count therefore satisfies
\begin{equation}
    Q_{\mathcal A}(g,T,\xi,\epsilon)
    \le3^{m+1}
    \le3\max\left\{1,\frac{\alpha(T)}{\alpha_0}\right\}.
    \label{eq:single-qubit-GP-angle-ratio-queries}
\end{equation}

For $0<u\le\pi/4$, we have
$\sin u\le u\le[\pi/(2\sqrt2)]\sin u$.
Using the definition of $\alpha(t)$ gives
\begin{equation}
    \frac{\alpha(T)}{\alpha_0}
    \le\frac{\pi}{2\sqrt2}
            \frac{\sin\alpha(T)}{\sin\alpha_0}
    =\frac{\pi}{2\sqrt2}
            \frac{\cosh a_0}{\cosh a(T)}.
\end{equation}
Since $a_0=a(T)+gT/2$ and
$\cosh(x+y)\le e^y\cosh x$ for all real $x$ and $y\ge0$,
we obtain
\begin{equation}
    \frac{\alpha(T)}{\alpha_0}
    \le\frac{\pi}{2\sqrt2}e^{gT/2}.
\end{equation}
The right-hand side is at least one, so it also bounds
$\max\{1,\alpha(T)/\alpha_0\}$.
Substituting this estimate into
Eq.~\eqref{eq:single-qubit-GP-angle-ratio-queries}
and using $e^{gT/2}\le2\cosh(gT/2)$ gives a bound
uniform over $0<\xi<1$ and all allowed preparation oracles.
Taking the supremum over $0<\xi<1$ proves
Eq.~\eqref{eq:single-qubit-GP-tighter-upper-bound}.

\end{document}